\documentclass[sigconf]{acmart}

\renewcommand\footnotetextcopyrightpermission[1]{} 
\usepackage{times}
\usepackage[utf8]{inputenc}
\usepackage{url}
\usepackage{array,multirow}
\usepackage{xspace}
\usepackage{xcolor}
\usepackage{color, colortbl} 
\usepackage{balance}
\usepackage{paralist}
\usepackage{cleveref}
\usepackage{graphicx}
\usepackage{pifont}
\usepackage{array}
\usepackage{booktabs}
\usepackage{tikz}
\usepackage{bm}
\usepackage{enumitem}
\usepackage{todonotes}
\usepackage{arydshln}
\usepackage{makecell}
\usepackage{tabularx}
\usepackage{caption}
\usepackage{soul,color}
\usepackage{subcaption}
\usepackage{mwe}
\usepackage{ifthen}
\usepackage{algorithm}
\usepackage[noend]{algpseudocode}

\newcommand{\cmark}{\ding{51}}%
\newcommand{\xmark}{\ding{55}}%

\newboolean{showcomments}
\setboolean{showcomments}{true}

\newcommand\important[1]{\todo[inline]{\textbf{Important:} #1}}
\newcommand\alberto[1]{\todo[color=yellow,inline]{\textbf{Alberto:} #1}}
\newcommand\etienne[1]{\todo[color=orange,inline]{\textbf{Etienne:} #1}}
\newcommand\pierre[1]{\todo[color=brown,inline]{\textbf{Pierre:} #1}}
\newcommand\mustafa[1]{\todo[color=blue!40,inline]{\textbf{Mustafa:} #1}}
\newcommand\michal[1]{\todo[color=green,inline]{\textbf{Michał:} #1}}
\newcommand\onur[1]{\todo[color=red,inline]{\textbf{Onur:} #1}}
\ifthenelse{\boolean{showcomments}} { }
{
\renewcommand\important[1]{}
\renewcommand\alberto[1]{}
\renewcommand\mustafa[1]{}
\renewcommand\michal[1]{}
\renewcommand\onur[1]{}
\renewcommand\etienne[1]{}
\renewcommand\pierre[1]{}
}

\ifthenelse{\boolean{showcomments}}
{ \newcommand{\mynote}[3]{
    \protect\fbox{\bfseries\sffamily\scriptsize#1}
    {\small$\blacktriangleright$\textsf{\emph{\color{#3}{#2}}}$\blacktriangleleft$}}}
{ \newcommand{\mynote}[3]{}}

\newcommand\sysname{Shard Scheduler\xspace}

\makeatletter
\def\BState{\State\hskip-\ALG@thistlm}
\makeatother

\newcommand{\etal}{\textit{et al.}\@\xspace}
\newcommand{\eg}{\textit{e.g.}\@\xspace}
\newcommand{\ie}{\textit{i.e.}\@\xspace}

\newcommand\para[1]{\vspace{0.05in} \noindent \textbf{#1.}}

\def\first{({\it i})\xspace}
\def\second{({\it ii})\xspace}

\definecolor{verylightgray}{gray}{0.8}

\newcolumntype{L}{l<{\hspace{1cm}}}
\newcolumntype{C}{c<{\hspace{1cm}}}
\newcolumntype{D}{c<{\hspace{0.3cm}}}

\begin{document}

\title[\sysname]{\sysname: object placement and migration in sharded account-based blockchains}

\author{Michał Król}
\affiliation{City, University of London}
\email{michal.krol@city.ac.uk} 

\author{Onur Ascigil}
\affiliation{University College London}
\email{o.ascigil@ucl.ac.uk} 

\author{Sergi Rene}
\affiliation{University College London}
\email{s.rene@ucl.ac.uk} 

\author{Alberto Sonnino}
\affiliation{Facebook Novi}
\email{asonnino@fb.com} 

\author{Mustafa Al-Bassam}
\affiliation{LazyLedger}
\email{mustafa@lazyledger.io} 

\author{Etienne Rivière}
\affiliation{UCLouvain}
\email{etienne.riviere@uclouvain.be} 

\renewcommand{\shortauthors}{M. Król, O. Ascigil, S. Rene, A. Sonnino, M. Al-Bassam, and E. Rivière}

\begin{abstract}
We propose \sysname, a system for object placement and migration in account-based sharded blockchains. Our system calculates optimal placement and decides of object migrations across shards and supports complex multi-account transactions caused by smart contracts. Placement and migration decisions made by \sysname are fully deterministic, verifiable, and can be made part of the consensus protocol. \sysname reduces the number of costly cross-shard transactions, ensures balanced load distribution and maximizes the number of processed transactions for the blockchain as a whole. It leverages a novel incentive model motivating miners to maximize the global throughput of the entire blockchain rather than the throughput of a specific shard.
\sysname reduces the number of costly cross-shard transactions by half in our simulations, ensuring equal load and increasing the throughput 3 fold when using 60 shards. We also implement and evaluate \sysname on Chainspace, more than doubling its throughput and reducing user-perceived latency by 70\% when using 10 shards.
\end{abstract}

\maketitle


\section{Introduction}
\label{sec:introduction}

Sharding emerged as one of the most promising layer-1 solutions to the scalability problems of blockchains~\cite{rapidchain, kokoris2018omniledger, luu2016secure, chainspace, wang2019monoxide, ethereum2}. A sharded system divides the blockchain infrastructure into groups called shards. Each shard has its own miners, holds a subset of the state, and processes a subset of transactions. This technique has the potential to increase the number of processed transactions per second, as they can be verified and agreed on in parallel by independent groups of miners. In theory, by increasing the number of shards, we can increase the global throughput of the blockchain.

A sharded blockchain~\cite{wang2019sok} can be seen as a distributed database where each transaction performs write operations, creating, destroying or modifying objects in one or multiple partitions (shards). We can distinguish between transactions writing to only one shard (intra-shard transactions) or to multiple shards (cross-shard transactions). Intra-shard transactions are relatively cheap and can be agreed on using the consensus protocol within their shard. In contrast, cross-shard transactions are more costly as they require local consensus in all involved shards as well as a cross-shard agreement between these shards. This is achieved using expensive techniques such as 2-phase commit~\cite{sonnino2019replay, chainspace, kokoris2018omniledger} or mutex-based protocols~\cite{rapidchain, ethereum2}. Finally, cross-shard transactions must be included in the chains of all shards holding involved accounts resulting in state inflation. The placement of objects in shards plays a crucial role in determining the overall performance (\ie the Transaction per Second--TPS--rate and the user-perceived confirmation latency).

In this paper, we focus on the account-based data model. Account-based objects are persistent. They represent user accounts (\ie user balance) or smart contracts and can be modified multiple times. Placing an object in a shard in the account-based model influences all future transactions for this object (in contrast to single-use transaction outputs in the UTXO model). 
Ethereum, the largest blockchain system supporting smart contracts, is an example of an account-based blockchain transitioning into a sharded mode of operation~\cite{ethereum2}.


Existing sharded blockchain designs generally use a static hash-based object-to-shard assignment~\cite{rapidchain, kokoris2018omniledger, luu2016secure, chainspace, wang2019monoxide, ethereum2}.
The hash space of object identifiers is divided equally between shards, and hashing the identifier of an object allows clients and miners to deterministically determine its location without using additional indexing services. In the long run, hash-based allocation equally spreads the load across shards but causes loss of data locality. Frequently interacting accounts may be spread across multiple shards causing costly cross-shard interactions~\cite{sok-consensus}. Furthermore, a fixed assignment cannot always react to activity bursts of accounts located in a single shard, causing short-term load imbalance. Both problems become more pronounced with an increasing number of shards and with an increasing number of accounts involved in each transaction, \eg as the result of the smart contracts executions.

\begin{figure}[t]
\centering
\includegraphics[scale=0.7]{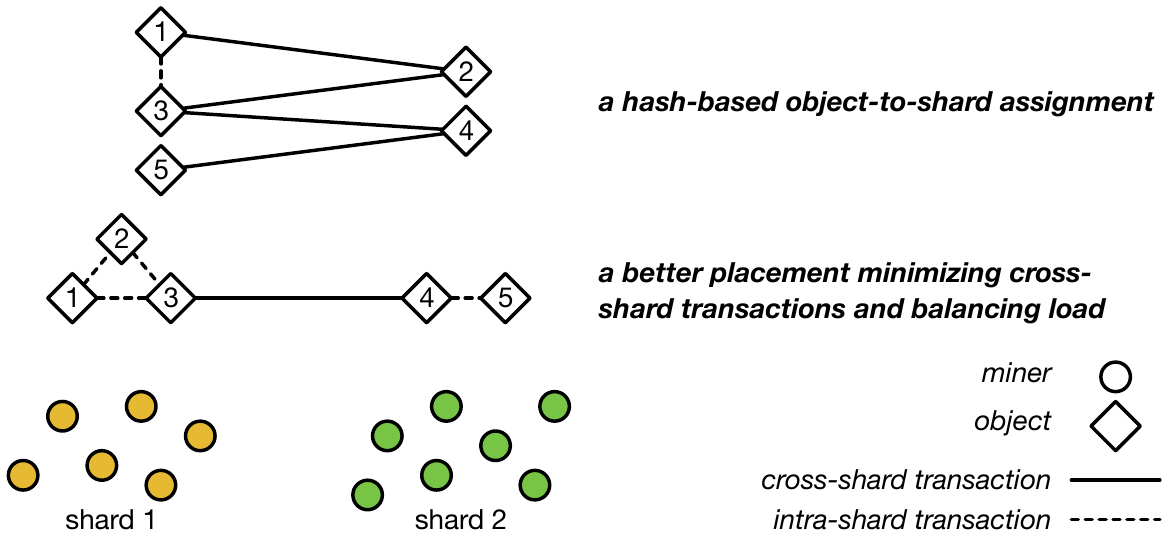}
\vspace{-5pt}
\caption{Object-to-shard assignment: a static placement (\eg hash-based) results in a high number of cross-shard transactions. A better placement could place (or migrate) objects 1, 2 and 3 in (to) shard~1 and objects 4 and 5 in (to) shard~2. 
}
\vspace{-15pt}
\label{fig:introduction}
\end{figure}

\Cref{fig:introduction} presents a simplified view of a blockchain with two shards and five accounts. Edges represent interactions (transactions) between accounts. The upper hash-based placement results in a high number of cross-shard transactions. 
A better placement is a compromise between load-balancing and the number of cross-shard transactions. We note that achieving such a placement through initial placement decisions only is not necessarily possible, and may require \emph{migrating} objects between shards (\eg accounts 2 and 5 in our example). Migration operations~\cite{fynn2018challenges, fynn2020smart, nguyen2019optchain} require additional transactions. The individual cost of these transaction executions, as well as the overhead they impose on the blockchain as a whole, must be worth paying, \ie result in higher throughput and lower confirmation latency for future transactions.


\para{Contributions} 
We present \sysname, a novel approach for deciding and enforcing object placement and migration decisions in sharded, account-based blockchains. Our scheduler balances the load between shards and improves data locality.
It leverages the possibility to initiate account migrations when necessary and seeks to maximize the global throughput of the blockchain.
At the same time, \sysname remains simple, deterministic, and verifiable for all the miners in the network to prevent abuse. \sysname is executed by the miners and does not require modifications of the clients, who are nonetheless able to verify the legitimacy of migration decisions taken as part of their transaction execution. Finally, \sysname makes scheduling decisions worth enacting for rationale miners through economic incentives.
We do not seek to propose novel mechanisms for handling cross-shard transactions and account migrations, but rather build upon the different proposals by other authors~\cite{rapidchain, kokoris2018omniledger, luu2016secure, chainspace, wang2019monoxide, ethereum2}. We only make minimal and common assumptions on the capabilities of the underlying sharded blockchain, allowing \sysname to be implemented on top of a vast range of account-based blockchains, from the upcoming evolution of Ethereum~\cite{ethereum2}, to current systems such as Zilliqa~\cite{team2017zilliqa}.

\para{Outline} 
We present a background on account-based blockchains and sharding mechanisms in \Cref{sec:background}.
We outline the design and perimeter of use of \sysname in \Cref{sec:overview}, present our assumptions on the underlying sharded blockchain together with our design goals in \Cref{sec:assumptions}, and present our system model in \Cref{sec:model}.
We then present our contributions as follows.

Our first contribution, presented in \Cref{sec:observations}, is an analysis of the transaction history from Ethereum from a perspective of a sharded execution. We use the Ethereum Virtual Machine (EVM) to extract all accounts that were modified by every transaction. We then investigate the activity of the accounts, their data locality, and the load balancing when using a static hash-based assignment.

In \Cref{sec:design} we present the design of \sysname, a transaction scheduler for sharded, account-based blockchains. \sysname observes system load and interactions between accounts to place and migrate accounts across shards to maximize the throughput.

In \Cref{sec:economics}, we develop and discuss an incentive scheme for sharded blockchains that motivates miners to maximize the TPS of the blockchain as a whole. By deploying this scheme, we free blockchain end-users from costly, manual migrations of the state and avoid associated security problems. Furthermore, we incentivize miners to perform migrations providing the highest global TPS instead of focusing on the fees collected on their own shard.

In \Cref{sec:evaluation}, we quantify the performance gain over a hash-based approach using a  simulator.

In \Cref{sec:prototype}, we present the integration of \sysname with the Chainspace~\cite{chainspace} sharded blockchain system and the results of its deployment on a large-scale testbed. Our evaluation shows that \sysname can adapt to many potential configurations of a sharded environment, more than doubles the throughput of the system, and lowers the latency by 65\% for 60 shards. 

Finally, \Cref{sec:conclusion} presents an analysis of \sysname properties, discusses future work, and concludes the paper.


\section{Background}
\label{sec:background}

In this section, we present background on account-based blockchains. We then discuss their transition into a sharded mode of operation, cross-shard transactions and migrations. 

\subsection{Accounts, state and transactions}

A blockchain is an append-only ledger maintained by a number of nodes called miners. A blockchain is expanded by the addition of blocks by designated miners, who receive incentives for extending the chain with correct blocks and behave according to the protocol. A block consists of a block header together with a list of transactions. Transactions modify the state of the ledger ranging from simple coin transfers to invocation of sophisticated smart contracts. The block header contains a hash of the block, the hash of the previous block, the hash of the state snapshot at a given time, and additional information related to the consensus protocol. Each block has a fixed capacity limiting the number of transactions it can contain. Including a transaction in a block requires some of the available total capacity of the blockchain system. We refer to the capacity required by a transaction as the \emph{cost} of that transaction. The cost usually depends on the size of the transaction (as done in Bitcoin~\cite{nakamoto2019bitcoin}) or its complexity (as done in Ethereum~\cite{wood2014ethereum}).

Miners that store all the blocks (including all the transactions) are called full nodes. In contrast, light nodes store only block headers and reactively pull required state elements or transactions from full nodes when needed. Light nodes can verify the integrity of the received data by comparing its hash against the value in the corresponding block header (\ie using Merkle proofs~\cite{merkle1987digital}).

In the account-based data model, the state of a blockchain consists of a list of objects representing accounts and their respective states. An account is accessed by its identifier (\eg a hash of its owner's public key) and represents an externally owned account (EOA), or a contract account (CA). For EOAs, the state consists of their balance. For CAs, the state may include more complicated data structures related to the logic of a smart contract. Importantly, while the state of EOAs is small and does not grow in time, the state of CAs can inflate as more data is put in the storage.

The state of an account can be modified by two types of transactions: external and internal. A transaction is external if sent from an EOA. For instance, a coin transfer, a contract creation, and a contract invocation are the 3 main external operation types happening in Ethereum~\cite{chen2020understanding}. Alternatively, a transaction is internal if it results from executing a smart contract invoked by an external transaction. 
A single external transaction may lead to multiple internal transactions depending on the smart contract logic.

\begin{algorithm}[t!]
\begin{algorithmic}[1]
\Procedure{payAll}{}
\State $\textit{users} \gets \text{a list of users to be paid}$
\State $\textit{amount} \gets \text{amount to pay each account}$
\For{\textit{user} \textbf{in} \textit{users}}
\If {$\textit{user.balance} < 10$}
\State $\textit{user.transfer(amount)}$
\EndIf

\EndFor
\EndProcedure
\end{algorithmic}
\caption{Example of a smart contract function modifying the state of multiple accounts.
}\label{alg:contract}
\end{algorithm}

A regular account-based transaction (\ie a simple coin transfer) modifies the state of up to 2 EOAs (the balance of the sender and that of the receiver). With the addition of Smart Contracts, transactions can lead to the modification of multiple accounts. \Cref{alg:contract} presents a Smart Contract implementing a \emph{payAll()} function. Calling this function modifies the state of the caller (to pay the transaction fees), the smart contract (to decrease its balance), and all the accounts stored in the \emph{users} map (to increase their balance), provided they currently have less than 10 coins. Smart contracts can also interact with and modify the state of other contracts by invoking their functions. Processing smart contract transactions require the write and read sets to be known to the consensus protocol layer based on the current state of the blockchain. 

\subsection{Sharding}

In fully sharded environments\footnote{Fully sharded environments split both the state and the transaction processing. Some sharded blockchains such as Monoxide~\cite{wang2019monoxide} or Elastico~\cite{luu2016secure} only split the latter and do not fall into this category.}, the blockchain is split into multiple groups with their own chains of blocks and miners. Each shard maintains and modifies the state of only a subset of the accounts existing in the system. Objects to shards assignments are usually static unless changed in explicit migrations caused by miners or users. A migration locks (or destroys) an object in the source shard and recreates it in the destination shard using an atomic transaction. The object identifier may or may not change during the migration depending on the underlying objects-to-shards mapping system. 
Shards are expanded by running local consensus protocol between shard-specific miners. Some designs~\cite{rapidchain, ethereum2} use a main chain that is used for coordination. The main chain periodically assigns miners to shards to prevent malicious miners from freely migrating and taking over a specific shard. As a result, only miners assigned by the main chain have the right to participate in the intra-shard consensus~\cite{wang2019sok}. Furthermore, the main chain stores block headers of all the shards, which facilitates cross-shard communication~\cite{rapidchain, ethereum2}.

\para{Cross-shard communication and migrations}
Transactions modifying the state of accounts placed in a single shard can be processed using intra-shard consensus similarly as in a non-sharded scenario. If the involved accounts are spread across multiple shards, however, executing the transaction requires cross-shard consensus to ensure the atomicity of transactions. There are two main types of cross-shard consensus protocols, \first protocols based on a two-phase commit protocol~\cite{gray1978notes} such as S-BAC~\cite{chainspace} and Atomix~\cite{kokoris2018omniledger}, and \second mutex-based protocols such as RapidChain~\cite{rapidchain} and the upcoming version of Ethereum~\cite{ethereum2}. In all cases, a cross-shard transaction requires an intra-shard consensus run in each shard holding at least one of the involved accounts together with the run of cross-shard coordination. The latter always causes additional overhead in all the involved shards. If any of the shards involved rejects a transaction, all other shards should likewise reject it to guarantee atomicity; that is, an atomic commit protocol typically runs across all the concerned shards to ensure the transaction is accepted by all or none of those shards. It also means that the processing time of a cross-shard transaction is determined by the slowest shard.

Objects can be migrated across shards by users (in explicit cross-shard transactions~\cite{nguyen2019optchain}) or by miners (as a part of the consensus protocol~\cite{rapidchain}). Performing migrations cause processing overhead for the miners and transaction fees for the end-users. The cost of migrations can be reduced when combined with cross-shard transactions. If account $A$ in shard $1$ sends a transaction to account $B$ in shard $2$, both accounts may remain in their respective shards (causing a costly cross-shard consensus round) or one of the accounts can be migrated to the shard of the other one\footnote{Both accounts can be also migrated to a third or different shards. However, such migration would cause significant overhead to the system.}. In the latter case, the migration cost still needs to be paid, but further processing requires cheaper intra-shard consensus in the destination shard. 

The use of migration can have a significant impact on the performance of the account-based blockchains.
This impact can be positive or negative depending on the migration decisions made.
Splitting frequently interacting communities may negatively impact the throughput of the entire system for many future blocks. On the other hand, migrations can equally spread the load across shards on a per-block basis improving resource utilization. Migrations increase the cost of individual transactions but, if done correctly, can also bring long-term performance gains. Correctly incentivizing decisions that are good for the blockchain as a whole can significantly improve the throughput of the entire system. We further discuss the topic in \Cref{sec:economics}.




\section{Overview}
\label{sec:overview}

The goal of \sysname is to integrate smart, automatic account placement and migrations decisions to improve the throughput of the sharded blockchain \emph{as a whole}. Our system strikes a balance between balanced load distribution, data locality, and the number and costs of performed migrations. \sysname performs migrations that are supported by the underlying consensus protocol, introduce relatively low short-term overhead, and reduces the cost of future transactions in the long run.

A fundamental design principle of \sysname is the implementation of our system on miners as a part of the consensus protocol. While client-based migrations have been proposed for throughput improvements in the UTXO model~\cite{nguyen2019optchain}, such an approach is not effective for account-based blockchains. A transaction in the account-based model modifies the state of multiple accounts (\eg sender, receiver, smart contract) but is authorized only by its sender. It thus restricts potential migrations to moving the sender only. In contrast, migration decisions taken by miners as a part of transaction processing can achieve optimal placement by moving any account involved in a transaction.

In \sysname, all migration decisions are taken based on a state snapshot of the blockchain, are deterministic, and can be verified by other miners. With decision verifiability, \sysname protects against malicious miners who might attempt a denial of service attack by forcing sub-optimal migrations. Our system requires only simple arithmetic operations to take optimal migration decisions and introduces only negligible overhead to the transaction processing. 

\sysname decouples the mining process from the collection of fees and aligns rewards collected by the miners with the throughput of the entire blockchain, rather than with the performance of a single shard. Rational miners are thus incentivized to pay the overhead cost related to automatic migrations. Finally, \sysname is completely transparent for the clients submitting transactions to the blockchain and does not require any client-side modifications.

The performance of a distributed system is tightly coupled to its submitted workload. Before outlining our design, we analyze the transaction history of Ethereum together with state-dependent smart contracts calls and extract new insights that allow understanding expected cross-shard interaction dynamics and shape the design of \sysname.


\section{Assumptions and design goals}
\label{sec:assumptions}

We base our assumptions on Ethereum~\cite{wood2014ethereum, ethereum2}, the main account-based blockchain transitioning into a sharded environment with support for smart contracts. Where Ethereum does not yet specify all the design details of its transition to a sharded operation, we assume functionalities provided by academic sharded blockchains (Omniledger~\cite{omniledger}, Chainspace~\cite{chainspace}, and RapidChain~\cite{rapidchain}). The characteristics of these systems are shown in \Cref{tab:assumptions}.

\begin{table}[t]
\footnotesize
\newcolumntype{A}{>{\raggedright\let\newline\\\arraybackslash\hspace{0pt}}m{2cm} }
\newcolumntype{C}{>{\raggedright\let\newline\\\arraybackslash\hspace{0pt}}m{0.13\linewidth} }
\newcolumntype{D}{>{\raggedright\arraybackslash} m{0.30\linewidth} }
\begin{tabular}{Acccc}
\toprule
 & \textbf{Ethereum+} & \textbf{RapidChain} & \textbf{Chainspace} & \textbf{Omniledger}\\
\midrule
\textbf{Smart Contracts} & \cmark & \xmark & \cmark & \xmark \\
\textbf{Beacon Chain} & \cmark & \cmark & \xmark & \cmark \\
\textbf{Miners reshuffling} & \cmark & \cmark & - & \cmark\\
\textbf{Write set specified by transactions} & ? & \cmark & \cmark & \cmark\\
\bottomrule
\end{tabular}
\caption{\sysname assumptions in existing systems.}
\label{tab:assumptions}
\end{table}

\subsection{Security Assumptions}

We distinguish two types of actors:
\begin{itemize}
 \item \emph{users} are owners of EOAs that use the blockchain;
 \item \emph{miners} are maintainers of the blockchain.
\end{itemize}

We assume the presence of arbitrary malicious actors that can play the role of users or miners and try to disturb the system. No single user or miner is trusted by its peers. However, as for many sharded blockchain designs~\cite{chainspace, omniledger, byzcuit, rapidchain}, we assume that all shards have an honest consensus majority. With the current single-chain economic models applied to a sharded environment, miners may be incentivized to deviate from the protocol when taking migration decisions. In \Cref{sec:economics}, we develop an economic model for sharded blockchain that makes the honest majority assumption more probable in a real-world deployment.

We assume a partially synchronous network for 2PC-based protocols\footnote{This assumption is not required by the cross-shard consensus protocol \textit{per se}, but by the BFT protocol running within each shard.}~\cite{dwork1988consensus}, and a synchronous network for mutex-based protocols (in light of recent replay attacks against sharded blockchains~\cite{byzcuit}).

We assume a sharded blockchain environment as envisioned by Omniledger~\cite{omniledger}. A measure of time is determined from the chain length of an arbitrary shard and is divided into \emph{epochs} of equal length. In every epoch, nodes can manifest their intention to become miners for the next epoch by registering their public key to a dedicated smart contract~\cite{chainspace} (or hardcoded logic on a beacon chain~\cite{omniledger,ethereum2}). The system runs a black box Sybil detection algorithm (typically proof-of-work~\cite{nakamoto2019bitcoin,wood2014ethereum} or proof-of-stake~\cite{ouroboros}) that outputs the list of registered public keys of the nodes that will become miners during the next epoch. At the start of a new epoch, miners are shuffled and assigned to shards at random.

We assume the presence of a main chain (as in Omniledger~\cite{omniledger} and RapidChain~\cite{rapidchain}, and as proposed for Ethereum~\cite{ethereum2}) that stores the block headers of all the shards.
Each miner is a full client for its respective shard and acts as a light client for the beacon chain and all the other shards.
We assume the presence of a mapping service holding current accounts-to-shards assignments (\eg implemented as a Distributed Hash Table). 

Processing a cross-shard transaction requires modifying a set of objects. For a simple transfer transaction (\ie not a call to a smart contract) the set contains the sender and the receiver and can be read directly from the transaction data. For blockchains supporting smart contracts, the list of involved accounts for a specific execution may depend on the current state across multiple shards. In \Cref{alg:contract} for instance, the caller, the contract, and accounts from \emph{users} may be spread across multiple shards. The required state and a list of involved objects can be either proactively locked and provided to the miners by the user as part of the transaction data (as done in Chainspace~\cite{chainspace}) or reactively pulled by miners executing the transactions (as discussed for Ethereum~\cite{ethereum2}). 

Our system is orthogonal to the actual implementation of cross-shard transactions with or without smart contracts. For each transaction, \sysname relies only on a write set (\ie accounts whose state will be modified by this transaction). Such a set is already required to process smart contract transactions (\Cref{sec:background}). Finally, we assume that each cross-shard transaction is forwarded to a shard responsible for its execution. We refer to this shard and more specifically to a miner including the transaction in its block, as the \emph{transaction coordinator}. The transaction coordinator obtains a list of accounts to be modified and coordinates other shards involved in the transaction.

\subsection{Design Goals}
\label{sec:design_goals}


The design of \sysname targets the following properties.

\para{Migration and placement recommendations}
\sysname analyzes interactions between accounts and issues recommendations specifying how an incoming transaction should be handled and, in particular, what (if) migrations should happen.
These recommendations have the goal of keeping frequently interacting accounts within one shard while providing a balanced load across shards.
By reducing the number of cross-shard transactions and their associated overheads, and avoiding performance degradation due to overloaded shards, two goals participate in unison to an increased throughput (total number of transactions per second for a given capacity).

\para{Recommendation verifiability}
Each recommendation is deterministic and can be reliably verified by all other miners.
\sysname recommendations are part of the consensus and block validation protocols. This property is required to ensure the availability of the blockchain.
Without verifiability, malicious miners may attempt to move objects towards an overloaded shard or split frequently interacting communities, thus increasing the cost of transactions and lowering the number of transactions per second~\cite{mirkin2020bdos}.
Such a denial of service attack, even when targeting a single shard, influences the throughput of the entire blockchain due to the impact on cross-shard transactions.

\para{Lightweight recommendations} 
\sysname recommendations are generated on a per-transaction basis.
The system ensures that the amount of required computation is low and can be easily performed by all miners without introducing significant space and time overhead.
\sysname operations remain computationally tractable also when the number of accounts present in the blockchain grows.
\sysname does not introduce any significant network overhead (\ie fetching large, additional state from other shards).

\para{No changes for the clients}
\sysname is transparent for EOA owner and, in contrast to related work~\cite{nguyen2019optchain}, does not require additional operation or maintenance of state by users.

\para{Incentive model}
\sysname provides an incentive model for the miners to motivate them to follow the recommendations. The reward of each miner is proportional to the amount of performed work (\ie the number of mined blocks) and the total amount of rewards acquired by the blockchain as a whole.
Miners are still incentivized to compete for producing new blocks include a maximum amount of transactions.
However, miners do not benefit from keeping excessive numbers of accounts in their shards and ignoring ingoing or ongoing migrations recommended by \sysname.


\section{System Model and notation}
\label{sec:model}

We present the notations used throughout the rest of the paper, and the model in which \sysname operates. 
Notation are summarized in \Cref{tab:notations}. 

\begin{table}[]
\footnotesize
\newcolumntype{C}{>{\raggedright\let\newline\\\arraybackslash\hspace{0pt}}m{0.18\linewidth} }
\newcolumntype{D}{>{\raggedright\arraybackslash} m{0.25\linewidth} }
\begin{tabular}{CDCD}
\toprule
\multicolumn{4}{l}{\textbf{Parameters}}                        \\
\midrule
$s_i \in S$ & states &  $t_i \in T$ & transactions  \\
$o_i \in O$ & objects &  $acc_i \in ACC$ & accounts\\
$b_i$ & balance of $acc_i$ &  $\phi$ & mapping service\\
$c(t_i)$ & cost of $t_i$ & $C_i$ & capacity of $s_i$ \\
$m_{j \rightarrow k}(acc_i)$ & migrations & & \\
\bottomrule
\end{tabular}
\caption{Notations.}
\label{tab:notations}
\end{table}

\subsection{Blockchain Model} 

The blockchain is maintained by a number of miners $m \in M$ validating and processing transactions. We adopt a similar blockchain model as Al-Bassam \etal~\cite{fraudproofs}. We model the blockchain as a set of state variables that encode its state $s \in S$ and transactions $t \in T$; at any time $s \in S$ represents a snapshot of the state of every object (\ie accounts, smart contracts). The blockchain maintains an append-only log of ordered transactions $\{t_0...t_n\} \in T$. The blockchain starts in an initial state $s_0 \in T$ and transitions from one valid state to the valid next state with each valid transaction $t_i(s_i) \rightarrow s_{i+1}$.

\para{Sharded blockchains}
In sharded blockchains, nodes are divided into groups called shards $z \in Z$, and each shard maintains a subset of the objects. Shard $z_j$ at step $i$ maintains $s_{ij}: acc_k \in ACC_j$. We assume a service mapping objects to their respective shard $\phi(acc_i) \rightarrow z_j$ as defined by Chainspace~\cite{chainspace}. 

\para{Transactions lifecycle}
Each node holds all incoming transactions in a fixed-sized \emph{transaction pool} (also called \emph{mempool}). At every time step, the transaction pool of every node is completely filled with transactions from clients.
Executed transactions are removed from the transaction pool. Only valid transactions are considered (\eg, for coin transfers, both the sender and receiver exist and the sender has sufficient funds to make the transfer). Invalid transactions are discarded.

\subsection{Processing Capacity} 
\label{sec:processing-capacity}

The concept of \emph{processing capacity} is key to our model.
Every time period, each shard $z_i$ has a processing capacity $C_i$ indicating how many transactions it can process during that time period while maintaining a constant transaction latency. In practice, this capacity can be limited by a number of factors such a network conditions, the size of the shard, and specific implementations. We assume that each shard has the same capacity ($\forall{i,j} \; C_i = C_j$), and that the capacity of the whole blockchain is the sum of the capacity of all its shards $C = \sum C_i$.


The cost of cross-shard transactions is higher than the cost of intra-shard transaction; we denote $c(t_i)$ the cost of transaction $t_i$. The exact cost depend on the consensus protocol as well as the cross-shard agreement protocol. The cost of each cross- and intra-shard transaction depends also on its size $c(t_i) \propto \textit{size}(t_i)$. The larger the transaction, the longer it takes to propagate the information to all concerned nodes\footnote{We determine the exact cost of transactions for Chainspace~\cite{chainspace} in later sections.}.

To process a transaction, a shard needs to spend some of its capacity equal to the cost of the transaction. For an intra-shard transaction $t_i$, shard $i$ spends $c(t_i)$ and is left with a capacity $C_i = C_i - c(t_i)$. For an cross-shard transaction each concerned shard spends the cost of an cross-shard transaction; so the transaction can only be processed if all shards have enough capacity to process it during this time period.


\para{State migration}
\sysname migrates objects between shards. When object $o_i$ is migrated from shard $z_j$ to $z_k$, $m_{j \rightarrow k}(o_i)$ the shard assignment service $\phi$ is updated accordingly.
Similarly to transactions, state migrations also have a cost for all involved shard that depends on the size of the migrated object $c(m_{j \rightarrow k}(o_i)) \propto \textit{size}(o_i)$.

\section{Observations}
\label{sec:observations}

We start by investigating the transactions in the Ethereum blockchain from the perspective of a sharded operation.
Our observations motivate the design of \sysname.
For each transaction, we extract all the accounts whose state was modified.
Details on data extraction are presented in \Cref{sec:data_extraction}. 

\para{O1. Write-oriented}
In a blockchain, one can securely read the state from any honest participant. In contrast, writing to the blockchain is complex, because the data must be propagated to every single miner and agreed on using consensus protocol. In this work, we focus on writing state to the blockchain.

\para{O2. Hot Spots}
The activity of accounts can vary significantly (\Cref{fig:hotspots}). The top 20\% accounts (\eg popular exchanges) are responsible for over 92\% of overall transactions. In the context of sharding, the most active accounts should not all be placed in the same few (or unique) shard(s).

\begin{figure}[ht]
\centering
\includegraphics[width=\linewidth]{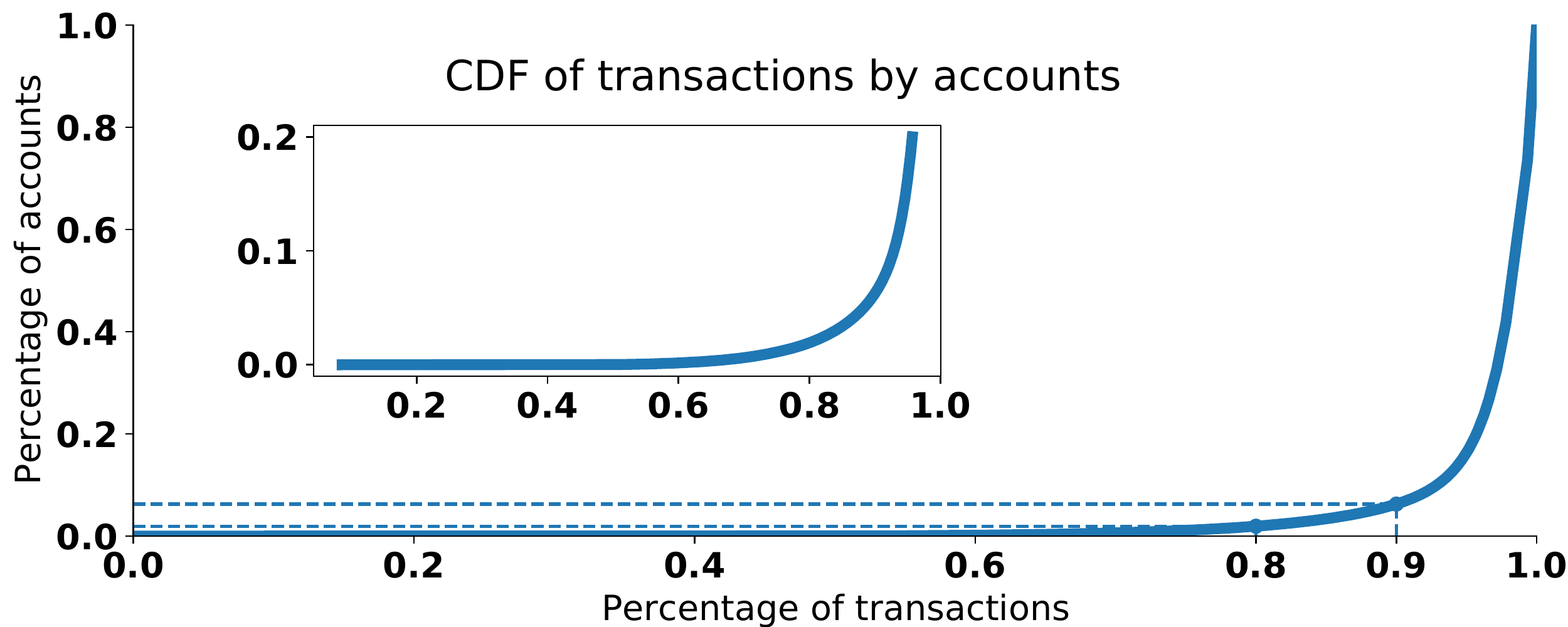}
\vspace{-15pt}
\caption{CDF of the number of transactions all the observed accounts were involved in.}
\vspace{-5pt}
\label{fig:hotspots}
\end{figure}

\para{O3. Communities}
Multiple works reported accounts forming communities, \ie groups of entities that interact frequently with each other~\cite{chen2020understanding, sun2019ethereum}. While the communities change over time, preserving them can significantly increase performance of a sharded blockchain due to the reduced number of cross-shard transactions~\cite{fynn2018challenges}. 



\para{O4. Load spikes}
To maximize the throughput of the system, each shard should utilize its full capacity. Accounts in Ethereum experience bursts of activities caused by the market (\eg Initial Coin Offerings, new tokens being added to exchanges) and ``follow the sun'' cyclical workload. We investigate the load shard of hash-based shard-to-account allocation (\Cref{fig:load}). We observe significant differences in shard load, especially for shorter periods of observation. Without account migrations, a sharded blockchain is not able to fully utilize its capacity. The problem becomes more pronounced with increasing number of shards. 

\begin{figure}[ht]
\centering
\includegraphics[width=\linewidth]{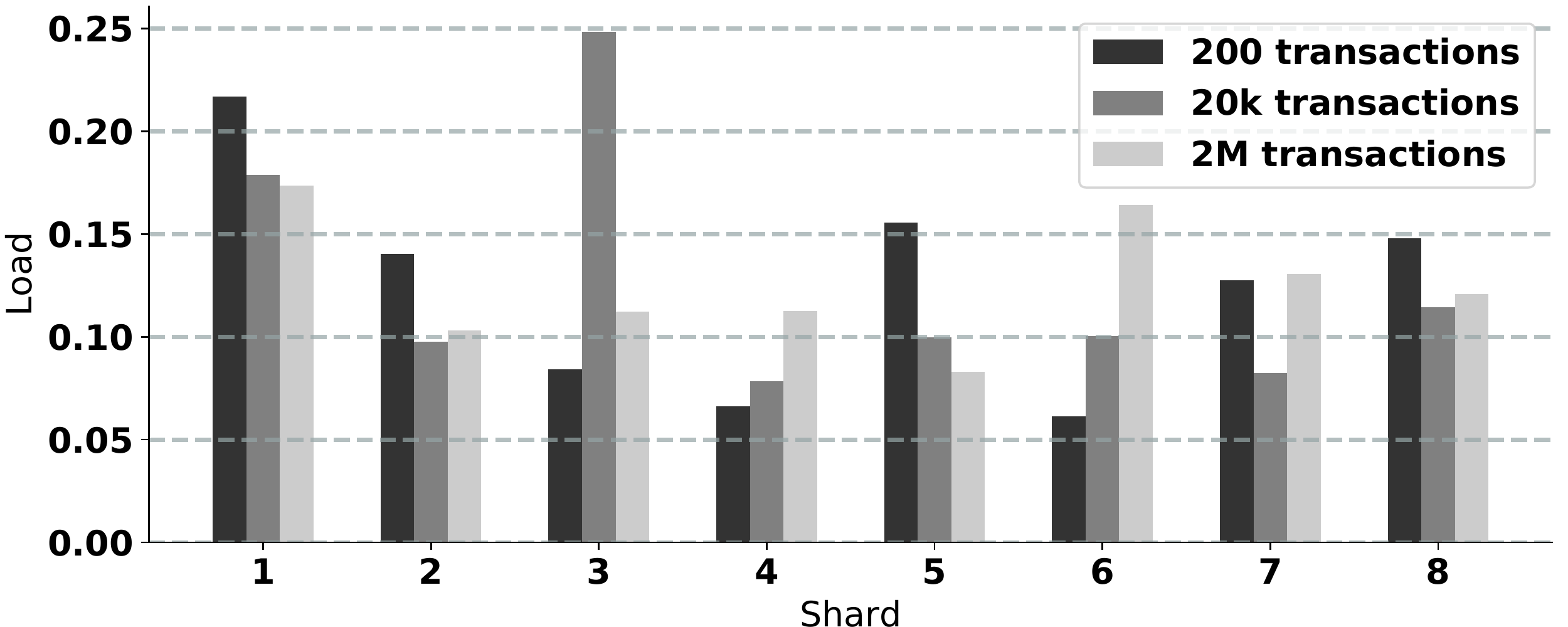}
\vspace{-15pt}
\caption{Shard load for hash-based account allocation for different periods of observation, when using 8 shards.}
\vspace{-5pt}
\label{fig:load}
\end{figure}

\para{O5. Migrating state during inter-shard transaction is cheap}
Under the model presented in \Cref{sec:model}, the cost of EOA migration is equivalent to the cost of an inter-shard transaction\footnote{A cost of of smart contract migration is proportional to its size.}. When two accounts spread across two shards are involved in a cross-shard transaction, one of the accounts can be migrated towards the other one replacing a cross-shard transaction by a migration and an intra-shard transaction. The intra-shard transaction will be processed only by the shard that hosts the accounts after the migration and does not generate additional overhead to the other shard.

\para{O6. Inactive accounts}
Accounts in blockchain are easy to create and are not constantly active. As of April 2020, the number of accounts exceeds 85~Millions, growing at a rate of about 50 to 150 thousands new accounts per day~\cite{etherscan}. However, only 3\% and 5\% accounts are active within one-week and one-month observation periods, respectively. A newly created addresses is used, on average, for 35.45 days before going inactive~\cite{ethereumreport}. At the same time, active accounts are likely to be updated soon after they are updated. An account is updated in a day from its previous activity with 62\% probability~\cite{kim2021ethanos}. We can say that only a fraction of accounts are active at any point of time, but once they are activated, they are likely to be accessed again soon (temporal locality). Inactive objects do not take part in new transactions and should not be migrated between shards even if they are highly connected with active objects. A migration of an inactive object involves a costly inter-shard agreement, does not decrease the state held by the input shard and increases the state held by the output shard without bringing any benefits. 

\para{O7. Smart Contracts}
%
Smart contract migration is a complex process~\cite{fynn2020smart}. A migration of a Smart Contract requires creating a snapshot of its current state, locking it in the input shard and its recreation in the output shard~\cite{fynn2020smart}. However, the process of creating the snapshot is complex and there are currently no efficient mechanisms to perform it. At the same time, the migration cost depends heavily on the size of the snapshot. In contrast to EOAs, the state size of smart contracts can be significant. 

\begin{figure}[ht]
\centering
\includegraphics[width=\linewidth]{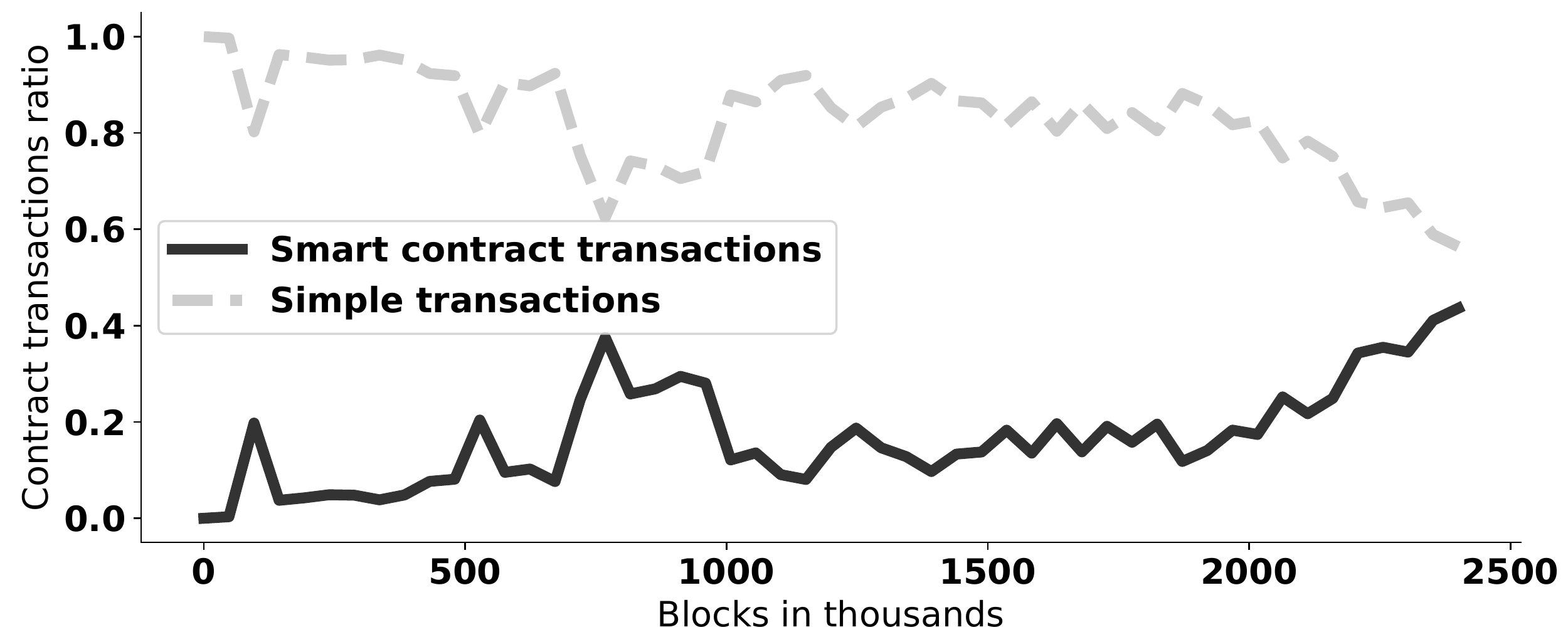}
\vspace{-15pt}
\caption{Simple and smart contract transactions over time.}
\vspace{-5pt}
\label{fig:internal_share}
\end{figure}
 
\para{O8. Internal Transactions}
Internal transactions constitute a growing part of all the transactions in Ethereum. \Cref{fig:internal_share}) presents the percentage of transactions per type. Ordinary user-to-user transfers are on a solid downward trend. Contract transactions take up to 45\% of the recent blocks in our sample. While majority of transactions modify the state of 2 accounts (EOA balances or internal state of smart contracts), some transactions modify up to 50 accounts at a time. Finally, the average number of accounts modified by an average transaction is on the rise over time caused by increased usage of smart contracts.


\section{\sysname design}
\label{sec:design}

\sysname is implemented as a part of the consensus protocol on all the miners of the blockchain. 
For each external cross-shard transaction, our scheduler operates in two steps:
\begin{enumerate}
    \item It determines the main shard and decide of placement of new accounts;
    \item It decides on the migration(s) of existing account(s) towards the main shard.
\end{enumerate}
We describe both steps in subsections below.
\sysname does not migrate any accounts not involved in pending transactions (\textbf{O5.~Migrating state during inter-shard transaction is cheap}) thus avoiding costly migrations that will not bring benefits in the future (\textbf{O6.~Inactive accounts}). The main shard selected during the first step is then used during the second step. Only the main shard will be considered as a potential migration destination.

\subsection{Data structures}

\sysname miners associate an alignment vector $$v_i = [a_{i1}, a_{i2}, ..., a_{in}]$$ with each account (including EOA and CA) in the blockchain where $a_{ij}$ represents the \emph{alignment} of account $i$ towards shard $j$. The \emph{alignment} is a positive integer and represents the total cost of transactions the account performed with the specific shard. When an account is created, the alignment vector values are set to $0$.  When account $acc_i$ in shard $z_i$ is involved in a transaction $t_k$ with account $acc_j$ in shard $z_j$, the respective values of both alignments vectors will be increased by the cost of $t_k$, so that $a_{ij} += c(t_i)$ and $a_{ji} += c(t_i)$. Importantly, $a_i$ will not be updated when $acc_j$ migrates between shards (and vice verse) simplifying the operation.

The alignment vector implements a sliding window approach and takes into account transactions from the last 100 blocks. It allows \sysname to better react to sudden burst of account activity (\textbf{O4.~Load spikes}) and reduce memory overhead, as empty vectors can be dropped from the memory. Due to the large number of inactive accounts (\textbf{O6.~Inactive accounts.}), \sysname maintains alignments vectors for a small fraction of the accounts at a time\footnote{We further show the memory overhead in \Cref{sec:evaluation}.}.

The alignment vector is held locally by each miner allocated to the shard where the account resides. It does not introduce any memory overhead to miners outside the shard and does not require storing any additional information on chain. The alignment vector is dropped (zeroed) when an account is being migrated between shards.

The second \sysname data structure is maintained on the beacon chain and represents the load of each shard in the system. The load for shard $z_i$ is a positive integer that holds total cost of transactions processed by the shard during last 100 blocks. Similarly to the alignment vector, implementing a sliding window approach improves \sysname reactivity to sudden load changes. The load is reported by shards when submitting their block headers to the beacon chain and is certified by the shard-specific miners. Placing the load information on the beacon chain, makes it available to all the miners in the system.

\subsection{Determining the main shard}

The first step is performed by the transaction coordinator. \sysname takes as input a new a list of accounts being modified by the transaction $l_i$, the shard assignment function $\phi$ and the last state of the blockchain $s_i$(as defined by the previous block on the beacon chain). The list is known to the coordinator and includes all the internal transaction caused by the external one (\textbf{O8.~Internal Transactions}). Based on this information, \sysname outputs allocation recommendations for new accounts (that appear on the blockchain for the first time) and a \emph{main shard} for the transaction. 

Based on the list of accounts and $\phi$, \sysname starts by creating a set of shards involved in the transaction. Consider the smart contract from \Cref{alg:contract}, and a transaction $t_j$ invoking the \emph{payAll} function. The list of accounts $l_j$ includes EOA of the caller, CA of the contract and accounts that from the \emph{users} mapping that have less than 10 coins\footnote{Assuming that the contract has enough money to pay all the accounts}.

If the set of shards involved in the transaction is not empty, \sysname then reads the load of each involved shard from the beacon chain and chooses the least loaded one as the \emph{main shard} for this transaction. If the set of shards involved in the transaction is empty\footnote{This may happen if the transaction modified the state of new accounts - e.g., a coinbase operation.}, our scheduler chooses the least loaded shard from all the shards.

\sysname assigns all the new accounts from $l_i$ to the main shard. The main shard identifier is then passed to shards holding non-new accounts involved in the transaction. The whole procedure for selecting the \emph{main shard} is illustrated by \Cref{alg:main_shard}.

\begin{algorithm}
  \caption{Main shard selection}\label{alg:main_shard}
  \begin{algorithmic}[1]
    \Procedure{selectMainShard}{$l_i, \phi, s_i,$}
      \State $\textit{involvedShards} \gets \text{set()}$
      \State $\textit{newAcc} \gets \texttt{new accounts from }l_i$
      \For{\texttt{acc} \textbf{in} $l_i$}
        \State \textit{involvedShards.add(}$\phi(acc)$\textit{)}
      \EndFor
    \If {\textit{involvedShards.empty()}}
    \State $\textit{mainShard} \gets \textit{lowestLoad(allShards)}$
    \Else
    \State $\textit{mainShard} \gets \textit{lowestLoad(involvedShards)}$
    \EndIf
    \For{\texttt{acc} \textbf{in} $newAcc$}
        \State $\phi(acc) \gets \textit{mainShard}$
      \EndFor
      \State \textbf{return} $mainShard$
    \EndProcedure
  \end{algorithmic}
\end{algorithm}

The main shard selection is based uniquely on the load shards. It allows \sysname to migrate accounts to the least loaded shard performing load balancing (\textbf{O2.~Hot Spots}).


\subsection{Deciding migrations of existing accounts}

The second step takes as input an account $acc_i$ involved in a cross-shard transaction, the shard assignment function $\phi$ and transaction-specific \emph{main shard} determined in the first step. The procedure is invoked only by miners associated with shard $z_i$, where $acc_j$ resides so that $\phi(acc_j) = z_j$. Importantly, the procedure does not require any external (from other shards) data and can be performed within the shard in question.

From the local state, \sysname consults account alignment vector and extract account's alignment towards its current shard and the main shard. If the alignment towards the current shards multiplied by the cost of the cross-shard transaction is smaller than the sum of alignments towards the other shards, the account will be migrated to the \emph{main shard}. Otherwise, the account remains in its current shard.

Taking into account the alignment vector stops the load balance-based migration if the account has strong connection with its current shard. Such an approach preserves existing cluster of frequently interacting accounts (\textbf{O3.~Communities}). The condition is more likely to stop the migration with increasing cost of cross-shard transaction. The whole procedure deciding on migrations is illustrated by \Cref{alg:migration}. For consensus protocols requiring all the accounts to reside in a single shard before processing (\ie mutex-based), \sysname always migrates all the involved accounts to the main shard.

\begin{algorithm}
  \caption{Migration decision algorithm.}\label{alg:migration}
  \begin{algorithmic}[1]
    \Procedure{shouldMigrate}{$acc_i, \phi, \textit{mainShard}$}
    \State $\textit{V} \gets \text{the alignment vector for~} acc_i$
    \If {\textit{(c(crossShard)V[acc])} < \textit{(sum(V) - V[acc])}}
    \State $\textit{migrate(}acc_i\textit{, mainShard)}$
    \EndIf
    \EndProcedure
  \end{algorithmic}
\end{algorithm}





\section{Economics}
\label{sec:economics}
Maintaining a blockchain requires resources to store (disk space), exchange (network bandwidth) and verify (CPU cycles) transactions. In open systems, miners are incentivized to perform this useful work in exchange for a financial reward. Incentive mechanisms for open sharded blockchains is currently a gap in the blockchain literature~\cite{sok-consensus}. We argue that naively applying incentives mechanisms from traditional (single-committee) blockchains to sharded systems has shortcomings, and then propose a novel design to fix them. 

\para{Purpose of the incentive mechanism}
The purpose of the incentive mechanism is to motivate rational miners to follow the protocol. In the absence of externalities (\eg secondary markets), it ensures that miners following the protocol collect a higher financial reward than if they were deviating from it. The main purpose of \sysname is to increase the performance of the blockchain. We require, therefore, an incentive mechanism that also goes in that direction: miners should be incentivized to follow the recommendations of \sysname.

\para{Traditional incentive mechanism}
Starting from Bitcoin, incentive schemes~\cite{nakamoto2019bitcoin,ethereum2,wood2014ethereum} typically involve collecting transaction fees from the end-user. The leader of the consensus protocol collects all the fees associated with the transactions it proposes; this leader is thus often rotated following system-specific strategies. 
Users are free to offer any fee for processing their transaction. Rational miners prioritize high fee transactions when constructing their blocks to maximize their financial reward. 

A naive extension of the incentive mechanism described above could work as follows. User associate transactions fees as in single-committee systems. These fees are shared amongst the leaders of the intra-shard consensus protocol of every shard involved in the transaction.
A similar incentive mechanism is adopted by Zilliqa~\cite{zilliqa}.

We argue that directly applying this mechanism to a sharded environment does not incentivize rational miners to maximize the system's performance. We show that, if given the right to perform account migrations, miners financially benefit from taking actions that harm the total system performance by creating imbalances between shards.


\begin{lemma}\label{le:greedy-miner}
In the sharded environment described in \Cref{sec:assumptions}, rational miners financially benefit from concentrating as many accounts as possible into their shard.
\end{lemma}
\begin{proof}
Miners are periodically elected as leaders according to the intra-shard consensus protocol and propose new blocks. When acting as the leader, rational miners choose the clients' transactions to include in their next proposal by selecting those with the highest fees. They can however only include transactions involving accounts in their own shards: these transactions are by definition a subset of the total transactions submitted to the system (for any epoch). As a result, miners have less options to select high-fees transactions than if they could choose amongst all transactions.
To increase the number of transactions that involve their shard, and thus increase their choice of transactions, miners are motivated to concentrate a large portion of accounts in their shard. 
\end{proof}
\Cref{le:greedy-miner} indicates that rational miners may financially benefit from actively resisting optimal placement recommendations, which may worsen the system performance.




\para{Adapting the model for sharded blockchains}
To overcome the shortcoming presented above, we propose an alternative solution that decouples the process of collecting transactions fees from cashing them in. We leverage the fact that miners are randomly assigned to shards, and thus cannot predict which shard they will integrate next (\Cref{sec:assumptions}). 
The incentive mechanism operates across every two consecutive epochs:
\begin{itemize}
    \item During epoch $n$, miners collect the fees of transactions that involve their shard and lock them into a shard-special deposit (as opposed to adding them to their private accounts). This deposit keeps a fine-grained accounting of the fees that each miner of the shard collected during the epoch. We follow classic incentive mechanisms and attribute the transactions fees to the current leader of the consensus protocol.
    \item Upon epoch change, miners are randomly shuffled and re-assigned to other shards. Upon entering the next epoch ($n+1$), miners cash in the transaction fees deposited into their new shard's deposit, as a pro-rato of their contribution during the previous epoch.
\end{itemize} 



\begin{figure}[ht]
\centering
\includegraphics[scale=0.7]{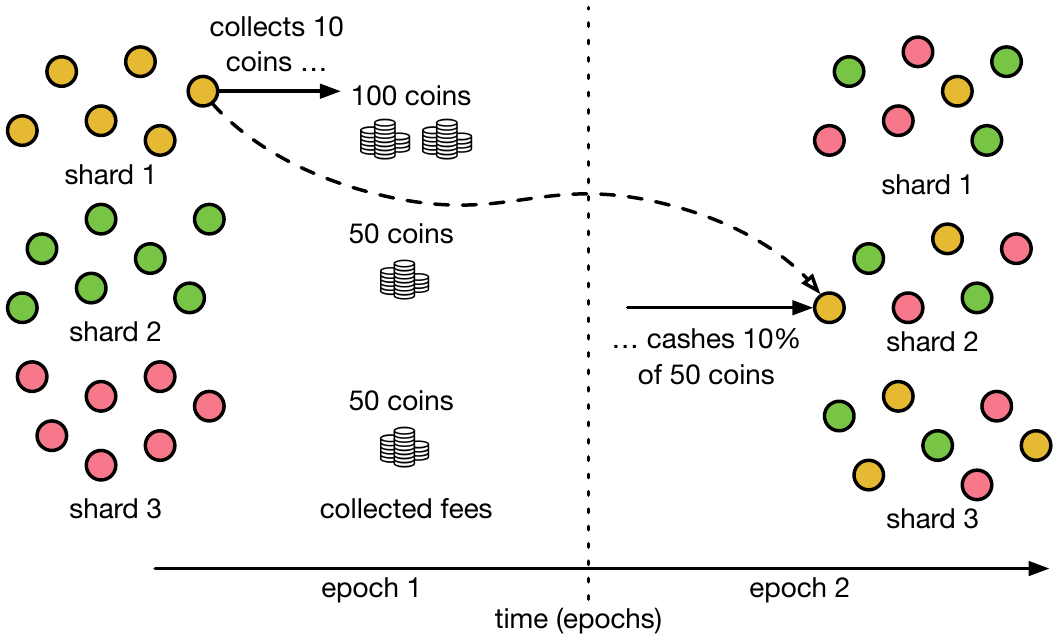}
\caption{Incentive model for \sysname.}
\label{fig:economics}
\end{figure}

Consider a scenario with 3 shards: $z_1$, $z_2$ and $z_3$ and miner $m_1 \in z_i$ for epoch $e_n$ (\Cref{fig:economics}). During epoch $n$, shard $z_1$ collects a total of 100 coins in transaction fees, $z_2$ collects 50 coins, and $z_3$ collects 50 coins as well. These fees are locked in their respective shard's deposit; that is, the deposit of $z_1$ holds 100 coins, and the deposits of $z_2$ and $z_3$ hold 50 coins each. No miners have access to these deposits for the time being.
Let's say that miner $m_1$, when acting as leader, proposed transactions containing a total of 10 coins of fees during epoch $n$. That is, we attribute $10\%$ of the total fees collected by $z_1$ during $e_n$ to miner $m_1$. 
During epoch $n+1$, $m_1$ is assigned to shard $z_2$. Upon entering the epoch, it cashes in $10\%$ of the deposit accumulated by $z_2$ during epoch $n$. That is, $m_1$ cashes in 5 coins.

\para{Effectiveness analysis}
We argue that our proposed incentive scheme incentivizes rational miners to increase the total system's capacity.

\begin{lemma} \label{sec:good-miner-1}
Each epoch $n$, the expected reward of miners is proportional to the total transaction fees collected in the system $x^{n-1}_{tot}$  during the previous epoch.
\end{lemma}
\begin{proof}
The expected reward of a miner during epoch $n$ is $E^{n}(x) = \sum_{i=1}^{k} x_i^{n-1} p_i$, where $x_i^{n-1}$ is the total reward collected by shard $i$ during epoch $n-1$, $p_i$ is the probability that the miner ends up in shard $i$ in epoch $n$, and $k$ is the total number of shards in the system. Since miners are randomly assigned to shards, $\forall i,j, \;\; p_i = p_j = \frac{1}{k}$. Thus $E^n(x) = \frac{1}{k} \sum_{i=1}^{k} x_i^{n-1} = \frac{1}{k} x_{tot}^{n-1}$.
\end{proof}

\begin{lemma} \label{sec:good-miner-2}
The total fees collected in the system $x_{tot}$ increases with the total capacity of the system $C$.
\end{lemma}
\begin{proof}
As described in \Cref{sec:model}, we assume that the shards' processing capacity $C_i$ is a scarce resource and that clients transactions are abundant. As a result, if the shards' capacity increases, they can process more transactions per epoch and thus collect more fees. This implies that the fees $x_i$ collected by shard $i$ increases with the shards' capacity $C_i$.
We can thus express $x_i$ in terms of $C_i$ as a monotonically increasing function: $x_i(C_i)$.

The total fee collected in the system is defined as $x_{tot} = \sum_{i=1}^k x_i$. We can thus write $x_{tot} = \sum_{i=1}^k x_i(C_i)$ to show that the total fees $x_{tot}$ increases with the shards' capacity $C_i$.

\Cref{sec:processing-capacity} defines the total capacity of the system as the sum of the capacity of every shard: $C = \sum_{i=1}^k C_i$, which means $C$ increases with $\{C_i\}_{i=1}^k$. Combining those observations, we have that both $x_{tot}$ and $C$ increase with the shards' capacity $\{C_i\}_{i=1}^k$. It follows that the total collected fees $x_{tot}$ increases with the total system's capacity $C$: $x_{tot}$ and $C$ are positively correlated. 
\end{proof}

\begin{lemma}
The expected reward of miners increases with the total system's capacity.
\end{lemma}
\begin{proof}
\Cref{sec:good-miner-1} implies that the expected reward of miners increases with the total fees collected in the system. \Cref{sec:good-miner-2} shows that the total fees collected in the system increases with the total system's capacity. Therefore, the expected reward of miners increases with the total system's capacity.
\end{proof}



\section{Evaluation}
\label{sec:evaluation}
We provide details on our data set as well as the setup and results of our simulations. 

\subsection{Data Extraction}\label{sec:data_extraction}
We download first 2M blocks Ethereum transaction history (1 year). We extract 8M non-coinbase transactions and all the accounts that were modified during each transaction. We use openethereum v3.2.3\footnote{\url{https://openethereum.org/}} operating in archive mode that allows to recompute all the intermediary states of the blockchain. To extract the transactions and state modifications, we create a Python tool based on web3.py\footnote{\url{https://github.com/ethereum/web3.py}} that queries the client with \emph{trace\_replayTransaction} calls in \emph{stateDiff} mode. We made the code and the dataset publicly available to the scientific community \footnote{Omitted for submission.}. 

\subsection{Setup}
We implement a Python-based simulator to evaluate the effectiveness of our approach. The simulator closely follows the model presented in \Cref{sec:model}, operates in rounds and takes transactions (extracted in \Cref{sec:data_extraction}) as the input workload. Before the first round, the simulator fills up the mempool with transactions from the input workload, and in the beginning of each subsequent round, the simulator tops up the mempool from the input workload. 

The size of the mempool is fixed and set up using simulation parameters. The transactions are processed in the order of arrival by the blockchain. The policy being evaluated indicates placement and in the case of \sysname, migration of objects. Each transaction increases the load of one or multiple shards. A transaction can be processed in the current round, only if there is enough processing capacity left in all the involved shards. The unprocessed transactions will remain in the mempool and be processed during subsequent rounds. The simulator reports the following performance metrics:

\begin{itemize}
        \item \textbf{Throughput} - the global throughput of the entire blockchain in terms of the number of transactions per block. 
        \item \textbf{Latency} - the average elapsed time to complete the processing of the transactions in the workload. We measure the elapsed time to complete a transaction in terms of number of rounds (blocks), \ie from the round when a transaction is initially read until the round when a transaction is added to the blockchain. 
        \item \textbf{Wasted Capacity} - the load-balancing performance of the system in terms of \emph{residual capacities} of the shards summed over all the rounds. Residual capacity of the shards is sum of unused capacity of the shards at the end of a round. 
        \item \textbf{Cross-shard transaction ratio} - the percentage of transactions that involve accounts from multiple shards. For \sysname, each migration is accounted as a separate cross-shard transaction.

\end{itemize}

We compare \sysname against hash-based and metis policy. The hash-based policy represents an approach used in existing sharded blockchains~\cite{rapidchain, kokoris2018omniledger, chainspace, ethereum2}, assigning accounts to shards based on their identifier and does not perform migrations. The Metis policy is a \textit{hypothetical} one that reads all the transactions in the beginning of the first round, at once and proactively performs sharding using the well-known metis graph partitioning (a.k.a. community detection) algorithm~\cite{metis} on the transaction graph, whose nodes are individual accounts and edge weights indicate the number of transactions between the accounts~\cite{dhillon2005fast}. 

The metis algorithm computes a desired number of ``balanced'' partitions, each corresponding to a shard, on an input transaction graph---the objectives of partitioning are to minimise the total weight of inter-partition edges (\ie minimising inter-shard transactions) and to minimise variance across partitions in terms of their total of intra-partition edge weights (\ie achieving similar number of intra-shard transactions in each shard). We do not compare against UTXO solutions such as OptChain~\cite{nguyen2019optchain} due to data model incompatibility. 


We verify the impact of the following parameters on the performance:
\begin{itemize}
    \item \textbf{Number of shards} - we vary this parameter from 1 to 60 shards, and set its default value to 16 shards. Higher number of shards means increased processing capacity, but also more cross-shard interactions and load balancing challenges.
    \item \textbf{Cross-shard transaction costs} - we assume a fixed cost of all the cross-shard transactions and measure the impact of changing this cost from one (as costly as intra-shard transaction) to ten. The actual cost depends on the consensus protocol and its implementation. We set the default value to 2 observed for Chainspace in \Cref{sec:prototype}. This parameter also impacts the cost of migrations performed by \sysname. We do not migrate smart contract accounts due to difficult to determine migration cost. 
    \item \textbf{Shard processing capacity} - we investigate the impact of modifying the processing capacity of a single shard. We set the default value to 200 (\ie 200 intra shard transactions per block) as observed for Ethereum~\cite{etherscan}.
    \item \textbf{Mempool-to-capacity ratio} - we express the mempool size in terms of the ratio of processing capacity of the entire blockchain per block. The mempool size of the system in practice depends on the rate of transaction submissions (\ie rate of arrival to mempool buffer) and processing speed (\ie rate of departures) of the blockchain. 
    
\end{itemize}
In each experiment, we only vary one system parameter while the rest of them take their default values. We start by measuring the performance of all the policies in terms of throughput and latency and later explain the results by observing \emph{wasted capacity} and \emph{cross-shard transaction ratio}.


\subsection{Results}

\begin{figure}[h]
  \centering
  \includegraphics[width=\linewidth]{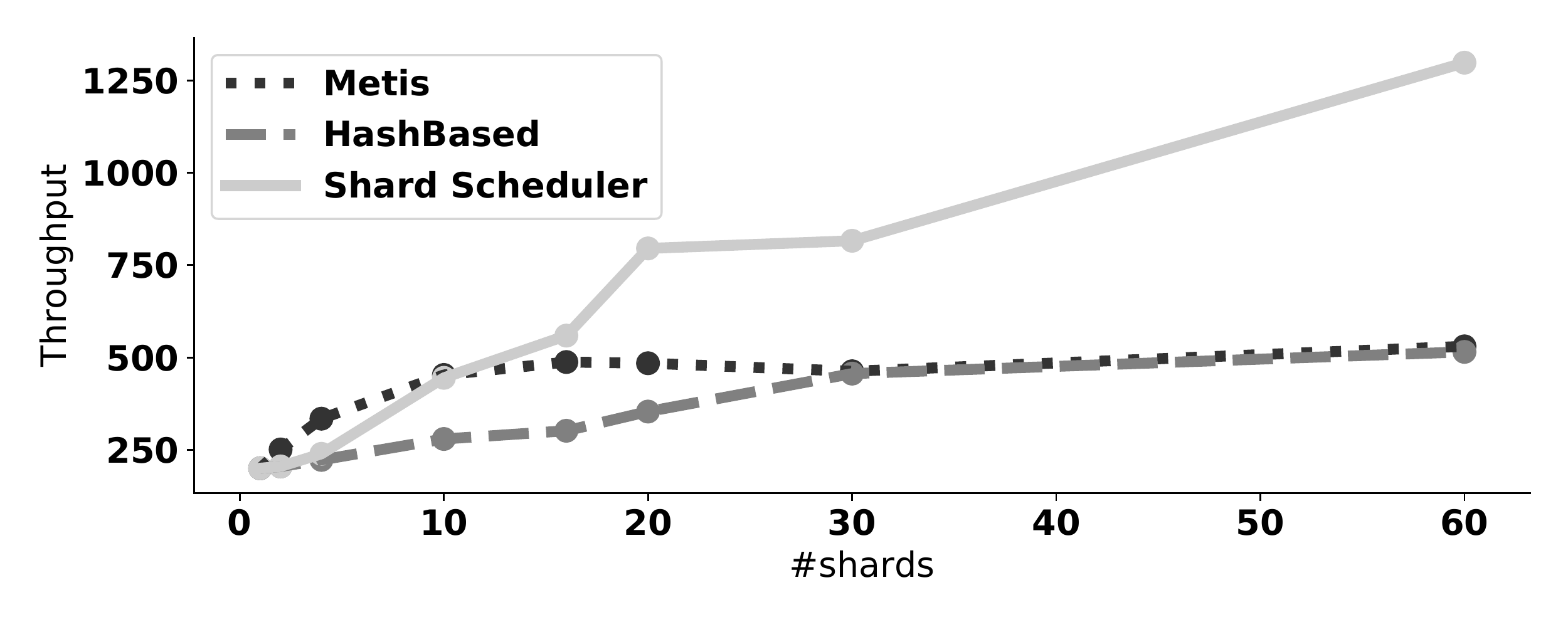}  
  \vspace{-25pt}
  \caption{Throughput vs number of shards.}
  \vspace{-5pt}
  \label{fig:sub_num_shards}
\end{figure}

\para{Throughput} \Cref{fig:sub_num_shards} shows the impact of the \textit{number of shards} on the throughput. We observe that \sysname achieves increasingly better throughput as the number of shards increases. On the other hand, the throughput of both metis and hash-based policies flatten out with increasing number of shards. \sysname improves the throughput by 100\% for 16 shards and by 250\% for 60 shard over the hash-based approach. \sysname also outperforms the theoretical metis policy, which uses future transaction information, for more than 10 shard and achieves similar performance for lower values.  

\begin{figure}[h]
  \centering
  \includegraphics[width=\linewidth]{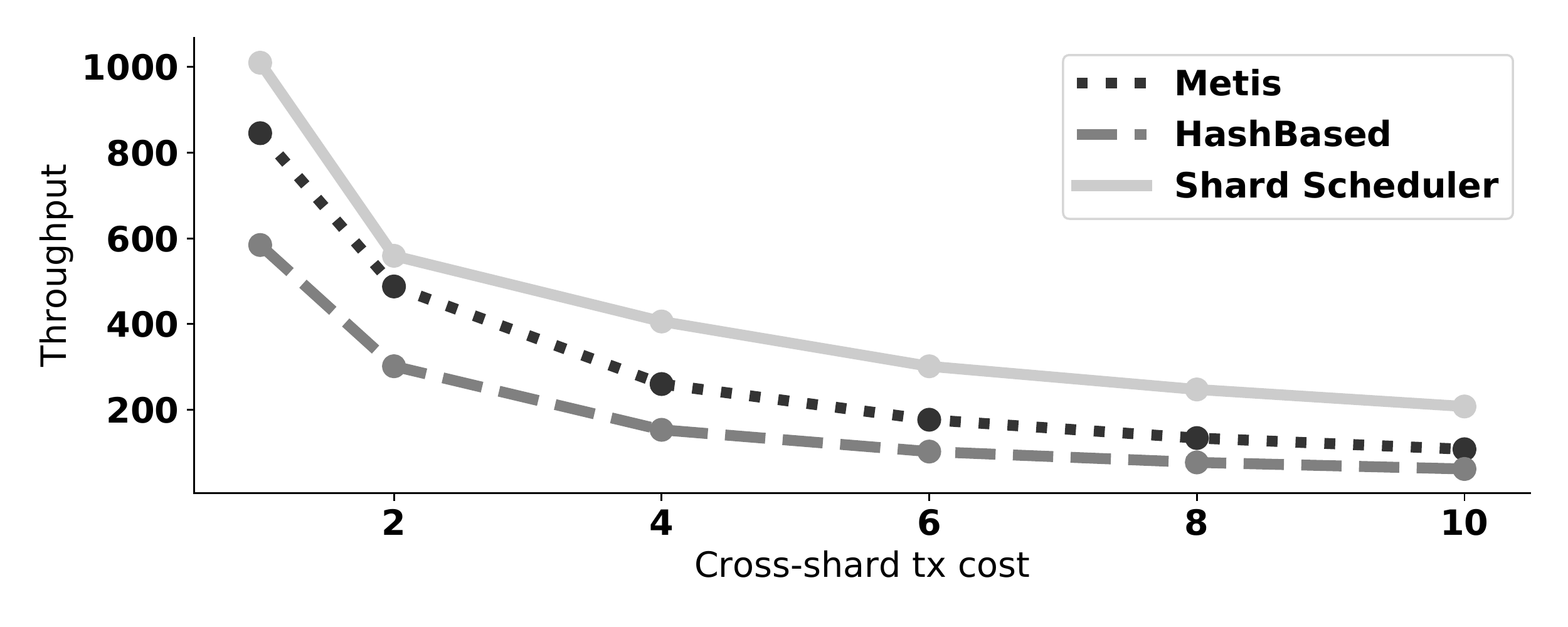}  
  \vspace{-25pt}
  \caption{Throughput vs cross-shard transaction cost.}
  \vspace{-5pt}
  \label{fig:sub_cross_shard_cost}
\end{figure}

\Cref{fig:sub_cross_shard_cost} shows the impact of varying \textit{cross-shard transaction costs} on the throughput. Higher processing cost results in lower throughput. For all the costs, \sysname achieves the highest throughput. The \sysname performance gain remains steady over both hash-based (80-95\% throughput increase) and  metis (10-40\% throughput increase) policies.

\begin{figure}[h]
  \centering
  \includegraphics[width=\linewidth]{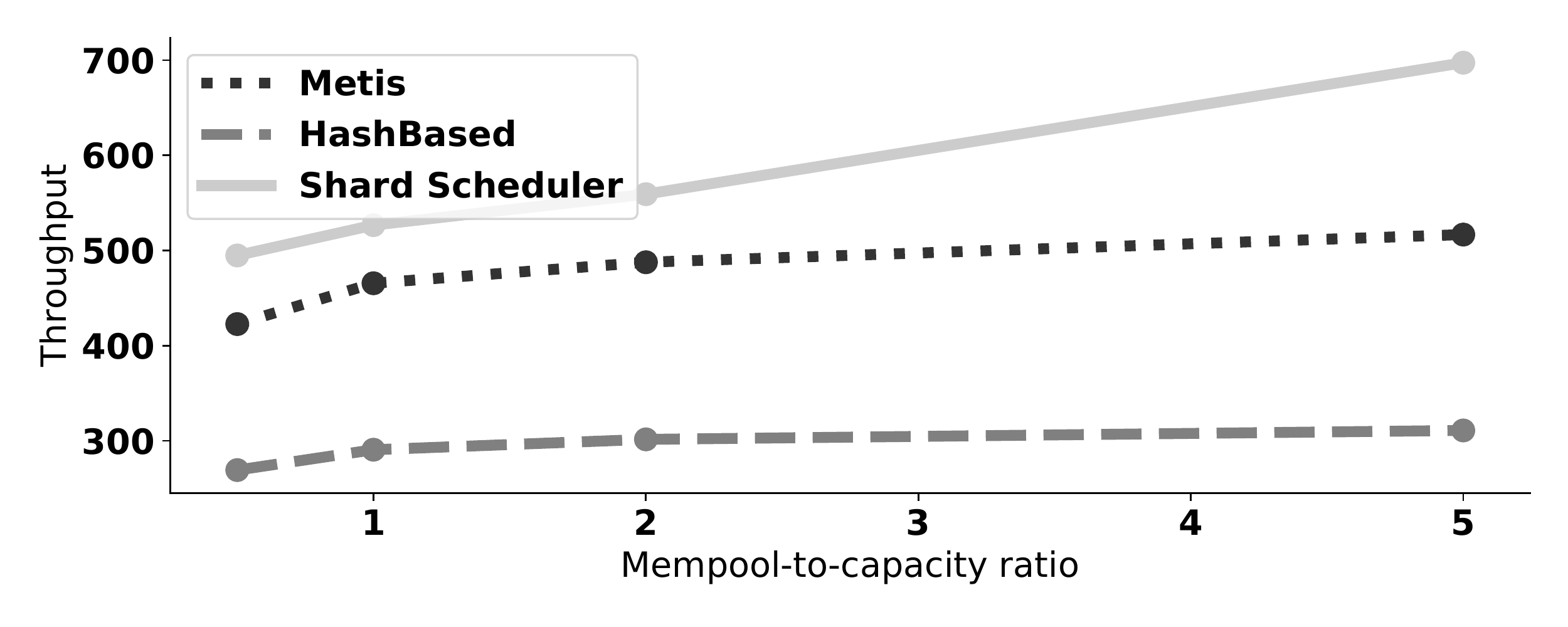} 
  \vspace{-25pt}
  \caption{Throughput vs mempool size.}
  \vspace{-5pt}
  \label{fig:sub_mempool_cap_ratio}
\end{figure}

In \Cref{fig:sub_mempool_cap_ratio}, we vary the \textit{mempool size} as multiples of shard processing capacity. Larger mempool size allows to achieve better load balancing and improve the throughput of all the policies. However, the impact of an increased mempoll size on hash-based and metis policies is limited. \sysname achieves 80\% throughput inceease for 0.5 mempool-to-capacity ratio and 130\% throughput increase for the ratio equal to 5.

\begin{figure}[h]
  \centering
  \includegraphics[width=\linewidth]{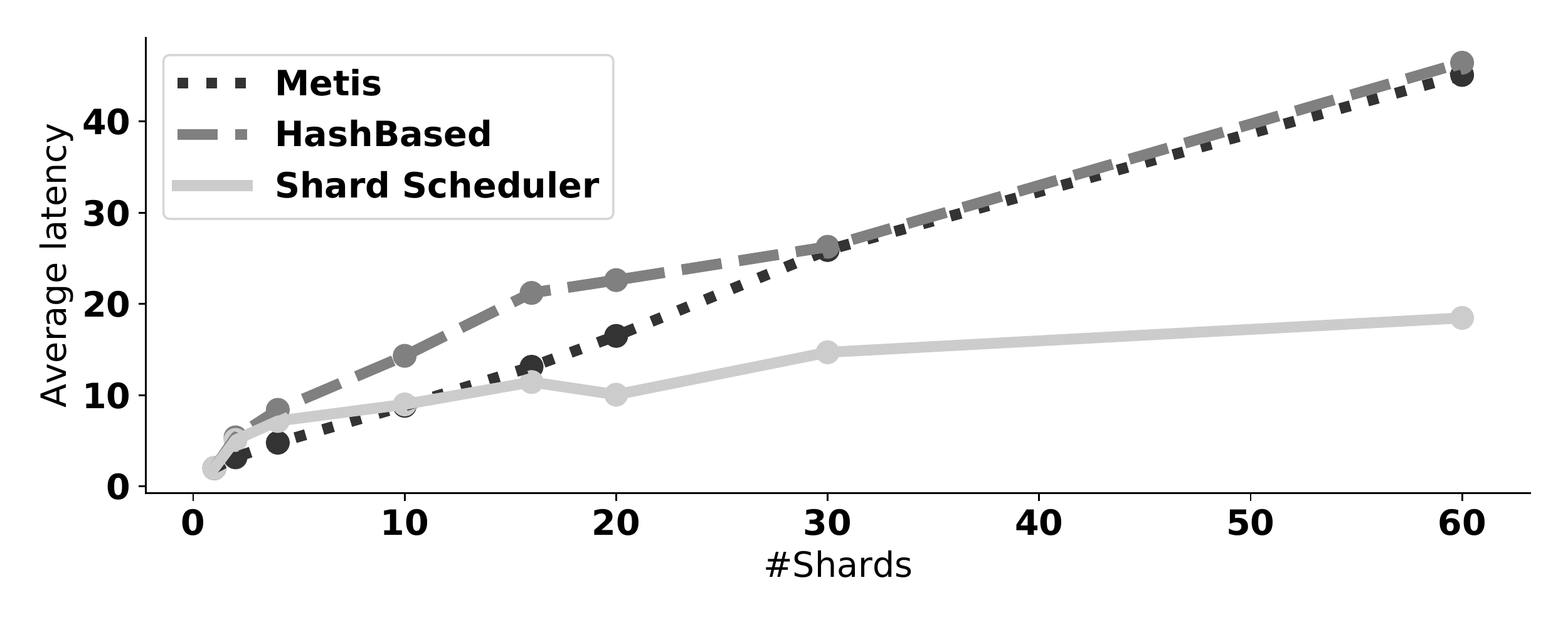}  
  \vspace{-25pt}
  \caption{Average latency vs number of shards.}
  \vspace{-5pt}
  \label{fig:sub_num_shardsl}
\end{figure}

\para{Latency} In ~\Cref{fig:sub_num_shardsl}, we observe average processing latencies with increasing number of shards. \sysname higher throughput result in significantly lower latency (3.5 times lower than other policies for 60 shards). Surprisingly high latency of the metis policy is caused by unequal load allocation (as shown below).

\begin{figure}[h]
  \centering
  \includegraphics[width=\linewidth]{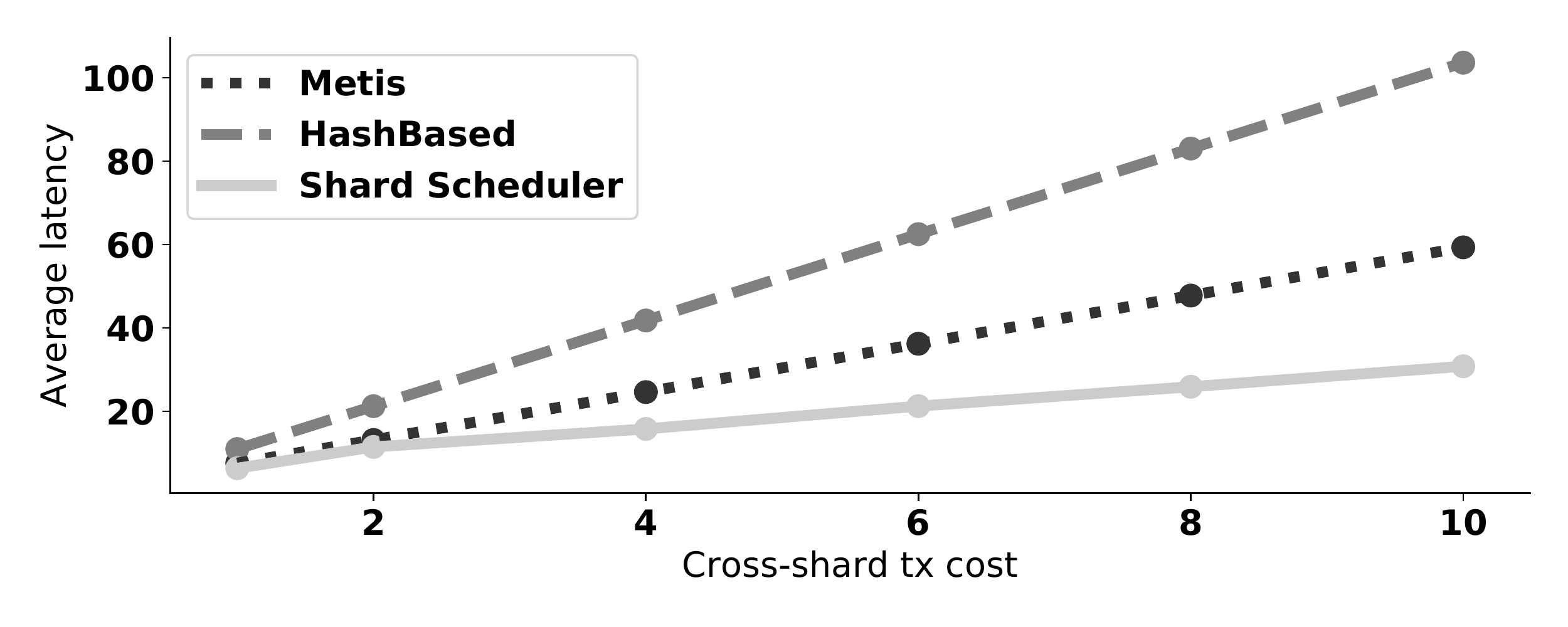} 
  \vspace{-25pt}
  \caption{Average latency vs cross-shard transaction cost}
  \vspace{-5pt}
  \label{fig:sub_cross_shard_costl}
\end{figure}
Increasing the cross-shard transaction cost (\Cref{fig:sub_cross_shard_costl}) increases the latency for all the policies. The metis policy preserves the account communities and performs better than the hash-base policy with the increasing cost of cross-shard interactions. However, \sysname achieves 2 times lower latency than metis and 3 times lower than hash-based policy for 10 cross-shard transaction cost.

\begin{figure}[h] 
  \centering
  \includegraphics[width=\linewidth]{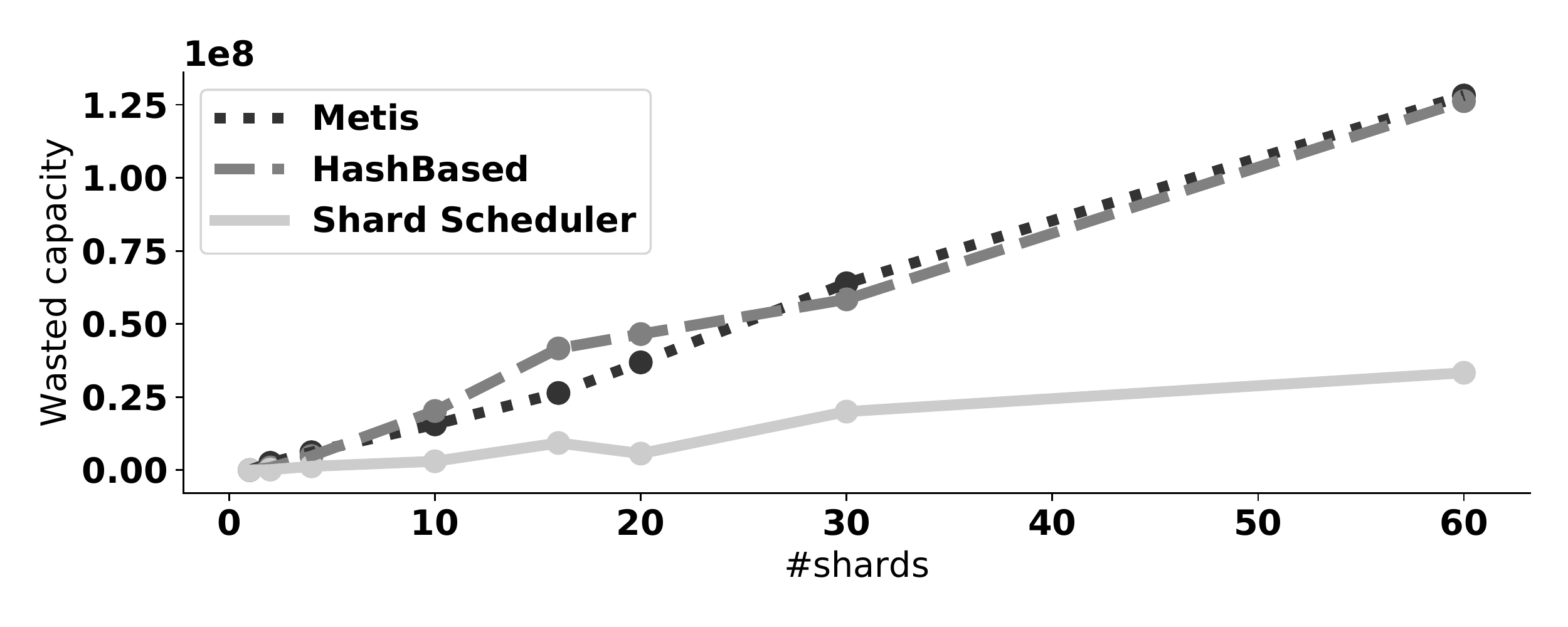}  
  \vspace{-25pt}
  \caption{Wasted capacity vs number of shards.}
  \vspace{-5pt}
  \label{fig:sub_num_shardsw}
\end{figure}

\para{Wasted Capacity} Both metis and hash-based policies achieve equal load spread across the shards in the long run. However, they fail to adapt to fine-grained activity changes due to the lack of migrations. \sysname takes per-transaction migration decision based on the previous load of all the shards and better utilizes the capacity of the blockchain. This effect is more pronounced as the number of shards (\Cref{fig:sub_num_shardsw}) or the cost of cross-shard transactions (\Cref{fig:sub_cross_shard_costw}) increases. More cross-shard interactions or their increased cost translates into more transactions waiting for one of the involved shards to become available.

\begin{figure}{}
  \centering
  \includegraphics[width=\linewidth]{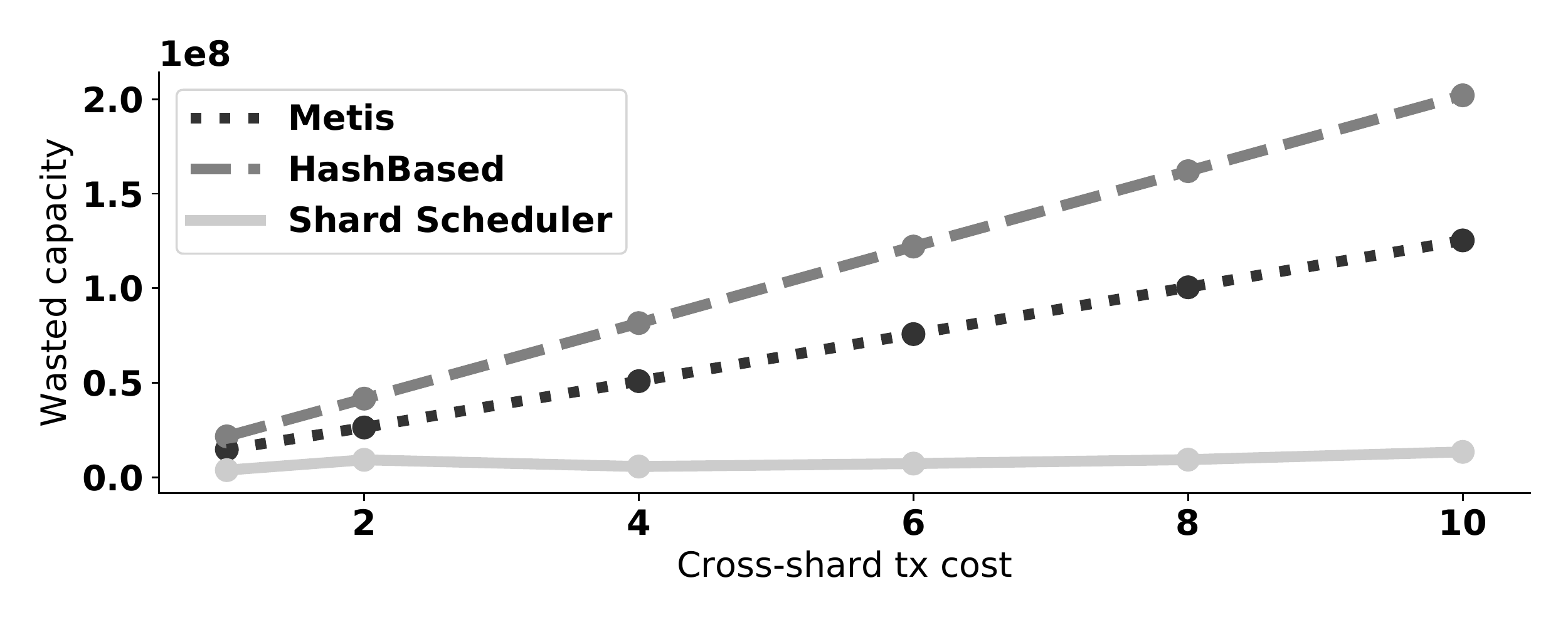}  
  \vspace{-25pt}
  \caption{Wasted capacity vs cross-shard transaction cost.}
  \vspace{-5pt}
  \label{fig:sub_cross_shard_costw}
\end{figure}

\begin{figure}[h]
\centering
  \includegraphics[width=\linewidth]{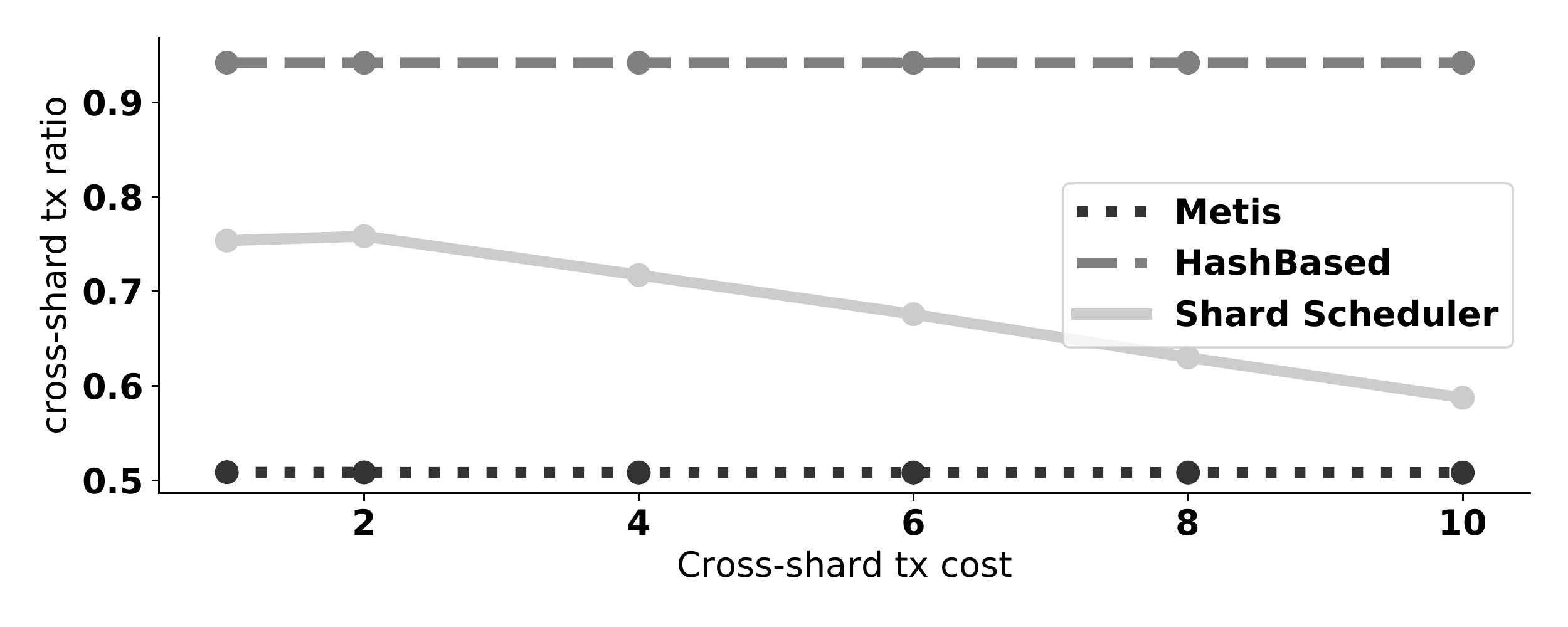} \vspace{-25pt}
  \caption{Cross-shard transaction ratio vs cross-shard transaction cost.}
  \vspace{-5pt}
  \label{fig:sub_shard_cap}
\end{figure}

\para{Cross-shard Transaction Ratio} Finally, we observe in \Cref{fig:sub_shard_cap} that \sysname is able to adapt gracefully to increasing inter-shard costs and reduce its number of cross-shard transaction ratio. This reduction is caused by the migration stopping condition (\Cref{alg:migration}) which takes the cross-shard transaction cost into account. On the other hand, both metis and hash policy are both oblivious to cross-shard transaction costs and their cross-shard ratios remain roughly constant failing to adapt to the changing environment. 

Overall, we observe that \sysname achieves significantly better performance despite scheduling additional cross-shard transactions linked to account migrations. The short-term migration overhead is largely compensated by long-term advantages of better load-balancing and preserving account communities.

\section{Prototype}\label{sec:prototype}
In this section we confirm the simulation results with real-world experiments. 

\para{Setup} We implement \sysname, metis and hash-based policied on top of Chainspace~\cite{chainspace} with security improvements proposed by Byzcuit~\cite{byzcuit}. Other shard environments are not yet finished~\cite{ethereum2}, do not provide the source code~\cite{wang2019monoxide, rapidchain} or does not fully partition the state~\cite{luu2016secure}. By default, Chainspace implements a UTXO data-model and does not implement blocks (\ie transactions are serialized as a continuous flow). We thus add blocks implementation and a data-model translation module that allows us to replay the history of Ethereum (\Cref{sec:data_extraction}) with an equivalent number of intra-shard and cross-shard transactions. We make the block implementation coherent with our model presented in \Cref{sec:model} and publish the code\footnote{Ommited for submission.}. We deploy 3 miners per shard on Amazon AWS within a single data centre and run tests for 5 and 10 shards. Due to high result variation within a single run, we repeat the tests 5 times and report the average values. 

We create 2 synthetic workloads of 1M transactions containing uniquely: \textit{(i)} intra-shard transactions and \textit{(ii)} cross-shard transactions. Both workload create perfectly balanced load across all shards. For the second workload, we observe 2 times lower throughput than for the first workload. We thus assume the cost of cross-shard transactions to be 2 for Chainspace and use it as a parameter to \sysname  (\Cref{sec:design}). 

\para{Results} We start by measuring the throughput of the system reported by Chainspace as the transaction per second (TPS) rate (\Cref{fig:chainspace_tps}). Surprisingly, we observe almost no throughput improvement for the hash-based policy when increasing the number of shards from $5$ (55TPS) to $10$ (56TPS). This is caused by highly unequal load across shards. For 10 shards, we observe multiple blocks filled to less than 50\% of their capacity. The metis policy provides much higher throughput (123TPS), but also suffer from unequal per-block load. The performance of the metis policy is expected to further drop down when increasing the size of the input file. \sysname is the only policy experiencing a significant throughput improvement and, for 10 shards, it triples the TPS rate of the hash based policy. However, the TPS rate when doubling the number of shards increases by only 23\%(from 120TPS to 148TPS). This is caused by migrations and cross-shard transactions that cannot be fully eliminated. Especially with complex smart contract transactions involving large numbers of accounts.

\begin{figure}
\centering
  \includegraphics[width=\linewidth]{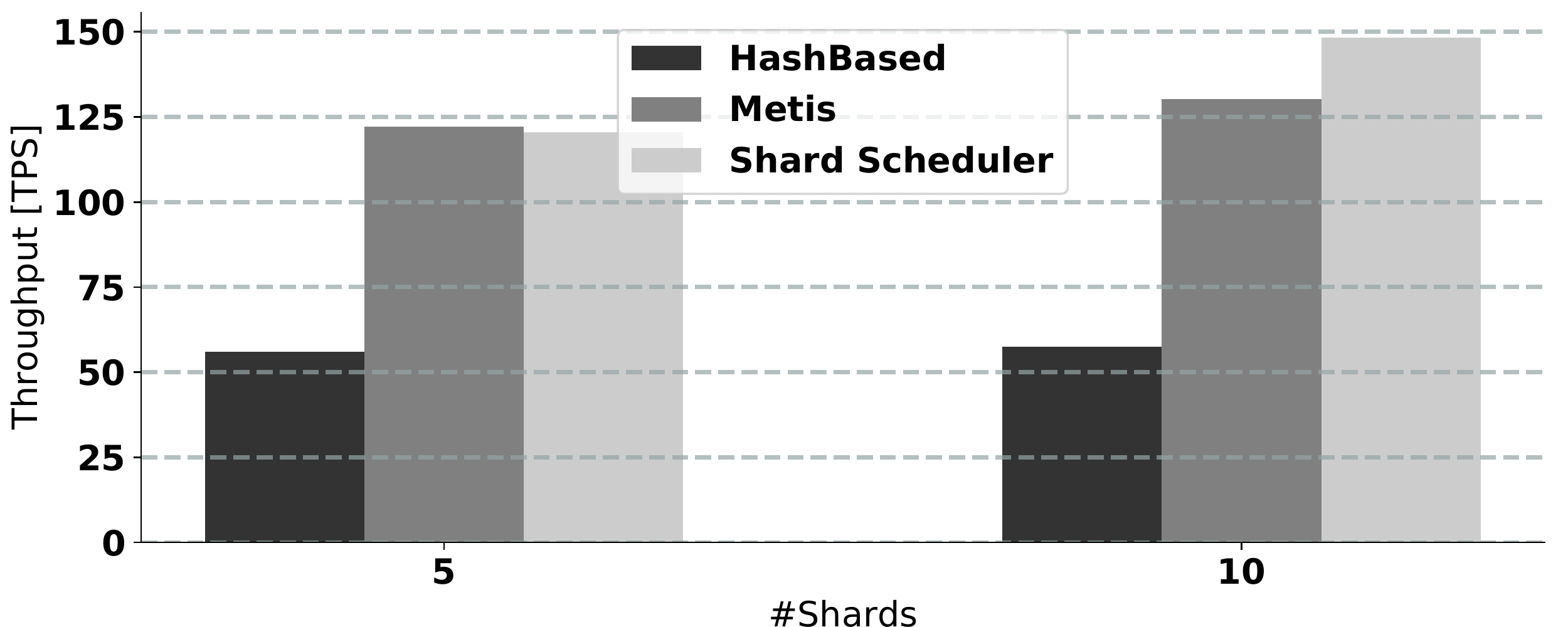}  
  \vspace{-15pt}
  \caption{Chainspace throughput.}
  \vspace{-5pt}
  \label{fig:chainspace_tps}
\end{figure}

We continue by investigating the transaction latency as perceived by the end-users (\Cref{fig:chainspace_latency}). Similarly to the simulations, the number of transactions submitted per block (\ie the mempoll) is proportional to the per-block capacity of the entire blockchain. Without the linear increase of the throughput, this approach causes increase of the user-perceived latency (as more blocks are need to fully process the mempoll). However, we observe the average latency achieved by \sysname to be significantly lower than for both the hash-based policy (49\% reduction for 10 shards) and the metis policy (31\% reduction for 10 shards). 

\begin{figure}
\centering
  \includegraphics[width=\linewidth]{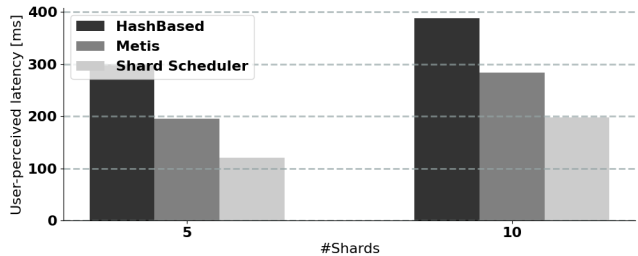}  
  \vspace{-15pt}
  \caption{Chainspace latency.}
  \vspace{-5pt}
  \label{fig:chainspace_latency}
\end{figure}



\section{Related Work}\label{sec:related}
We review related work on object migration and placement for sharded blockchain. We then briefly discuss object management techniques from the area of distributed systems. 

\para{Object migrations and allocation}
Optchain \cite{nguyen2019optchain} proposes an oracle for optimal transaction placement in sharded blockchains. 
The system uses graph clustering techniques and is implemented as an external service for the clients. 
However, Optchain approach works only for UTXO blockchains and cannot be easily adapted to the account-based data model. 
Han \textit{et al.}~\cite{hananalysing} study existing shard allocation protocols and propose WORMHOLE, a shard allocation protocol taking into account both self-balance and operability. 
However, the study focuses on allocating miners to shards, rather than objects residing on the blockchain. 
Fynn \textit{et al.}~\cite{fynn2018challenges} analyze the history of Ethereum transactions and investigate multiple graph clustering protocols in the context of account placement in sharding.
Similarly to our observations, they show that proactive placement without periodic migration does not achieve optimal performance. 
Fynn \textit{et al.}~\cite{fynn2020smart} develop techniques for moving smart contracts between shards and blockchains \textit{de facto} enabling contract migrations. 
The authors implements their protocol on Ethereum~\cite{wood2014ethereum} and Burrow~\cite{burrow}.

\para{Distributed systems}
In the area of the distributed systems multiple works investigated optimal object assignment and migrations. The proposed systems focus on two main aspects: \textit{(i)} developing a partitioning/migration plan (\ie object-to-partition allocation) and \textit{(ii)} efficient plan execution guaranteeing safety without causing significant downtime. 

E-store~\cite{taft2014store} provides an efficient solution based on tuples monitoring and bin backing problem to compute an optimal assignment of object to partition. However, the system does not take into account data locality. Clay~\cite{serafini2016clay} balances the number of inter-partition transactions, load balancing and limiting the number of migrations in order to maximize the throughput of the system. P-store~\cite{taft2018p} creates a partition plan taking into account load only. It contains a traffic prediction module~\cite{chen2008energy} that can proactively scale up or down the entire platform.

Squall~\cite{elmore2015squall} and Mgcrab~\cite{lin2019mgcrab} implement systems for efficient object partition and migration once given a partition plan. The platforms proposed for distributed systems provide important insights useful in our designs. However, they cannot be directly applied to sharded blockchains due to a different governance model. The majority of the platforms contain a non-deterministic element or cannot be verified by third parties~\cite{taft2018p}, introduce significant computational overhead~\cite{serafini2016clay, chen2008energy} or migrate large clusters of the objects at once~\cite{serafini2016clay}.



\section{Discussion and Conclusion}
\label{sec:conclusion}
\sysname provides objects migration and placement recommendations for account-based sharded blockchains. It provides a number of desirable properties and achieves the design goals specified in \Cref{sec:design_goals}.
First of all, \sysname improves the overall throughput of sharded blockchains. This is achieved through the mechanism explained in \Cref{sec:design} and its effectiveness is demonstrated by experiments (see Sections~\ref{sec:evaluation} and~\ref{sec:prototype}). 
Furthermore, \sysname recommendations are publicly verifiable as they are deterministic. Any third party can verify the correctness of objects migration and miners can apply the recommendations without needing an extra round of consensus. 
\sysname is lightweight in the sense that it does not require extra protocol messages, and does not introduce significant computation or memory overhead.
It integrates seamlessly into existing protocols requiring only minimal changes to the miners' software, and does not impact the way clients use the system.
\Cref{sec:economics} provides a novel incentive mechanism for sharded blockchain to financially motivate miners to maximize the total throughput of the system---miners collect higher fees by improving the overall performance of the system rather than by concentrating accounts in their own shard.

We leave a number of open questions that are deferred to future works. First of all, the objects placement recommendations of \sysname are efficient based on current and past typical usages of blockchains. There are no guarantees that this would be the case if blockchains are used in significantly different ways in the future. A learning agent may solve this issue by predicting future interactions between accounts, but it is not clear how to ensure that such agent remains both deterministic and lightweight. 
Secondly, handling transactions fees could become costly operations as they are associated with each transaction and may involve multiple shards. It would thus be desirable to remove fees handling from the critical path of the transaction's processing, or even offload them to a side-infrastructure. Recent works~\cite{fastpay}~\cite{astro} demonstrate that distributed payment systems can efficiently be implemented without consensus, and by quorum-based systems that can be natively integrated into one or more shards of a sharded blockchain.

\section*{Acknowledgements}
Alberto Sonnino is supported by Novi, a subsidiary of Facebook.
This work is partially supported by The Brussels Institute for Research and Innovation (Innoviris) under project FairBCaaS, and by the Belgian \emph{Fonds de la Recherche Scientifique} (FNRS) under Grant \#F452819F.

\bibliographystyle{plain}
\bibliography{references}

\end{document}